\documentclass[a4paper,12pt]{amsart}
\usepackage{amssymb}
\usepackage[colorlinks=true]{hyperref}
\usepackage{bm} 
\usepackage{tensor} 
\usepackage{tikz}
\usetikzlibrary{decorations.markings,calc,positioning}
\usepackage{nicefrac} 
\allowdisplaybreaks[1]
\title[Higher spin polynomial solutions of $q$KZ equation]{Higher spin polynomial solutions of quantum Knizhnik--Zamolodchikov equation}
\author[T.~Fonseca]{Tiago Fonseca}
\address{T.~Fonseca, LAPTh, Universit\'e de Savoie, CNRS UMR 5108, BP 110, 74941 Annecy-le-Vieux Cedex} 
\email{tiago.fonseca@lapp.in2p3.fr}
\thanks{TF was based at the Centre de Recherches Math\'ematiques at UdeM while most of the work was carried out. He is currently supported by ANR-10-BLAN-0120-03 ``DIADEMS''. Preprint LAPTh-060/12}%
\author[P.~Zinn-Justin]{Paul Zinn-Justin}
\address{P.~Zinn-Justin, UPMC Univ Paris 6, CNRS UMR 7589, LPTHE,
75252 Paris Cedex, France}
\email{pzinn@lpthe.jussieu.fr}
\thanks{PZJ is supported in part by ERC grant 278124 ``LIC''.
He would like to thank N.~Reshetikhin and R.~Weston for discussions.}
\newtheorem{theorem}{Theorem}[section]

\newtheorem{corollary}[theorem]{Corollary}
\newtheorem{lemma}[theorem]{Lemma}
\newtheorem{proposition}[theorem]{Proposition}
\newtheorem{remark}[theorem]{Remark}

\numberwithin{equation}{section} 
\long\def\rem#1{}
\def\eqbreak#1\\{\\*
\noalign{\vbox\bgroup\hfill$\def\eqbreak{$\egroup\nobreak\vbox\bgroup\hfill$}\displaystyle #1$\egroup}
}
\newcommand\poc[1]{\left(#1;q^4\right)_\infty}
\newcommand\normord[1]{{\,:\!#1\!:\,}}

\newcommand{\braket}[2]{\left\langle  #1 \vphantom{#2}\right|\!\left. #2 \vphantom{#1} \right\rangle}
\newcommand\braE[1]{v_{#1}^{\ast}}
\newcommand\ketE[1]{v_{#1}}
\newcommand\braN[1]{\tensor[_B]{\left\langle#1\right|}{}}
\newcommand\ketN[1]{\tensor[]{\left|#1\right\rangle}{_B}}
\newcommand\braC[1]{\left\langle#1\right|}
\newcommand\ketC[1]{\left|#1\right\rangle}
\tikzstyle{young}=[black, semithick]
\tikzstyle{dot}=[fill,shape=circle,black,inner sep=0pt,minimum size=3pt]
\tikzstyle{Pr}=[fill,shape=rectangle,red!60!white,inner sep=0pt,minimum size=8pt]
\tikzstyle{Ps}=[fill,shape=circle,blue!70!white,inner sep=0pt,minimum size=8pt]
\tikzstyle{pr}=[fill,shape=rectangle,red!60!white,inner sep=0pt,minimum size=4pt]
\tikzstyle{ps}=[fill,shape=circle,blue!70!white,inner sep=0pt,minimum size=4pt]
\tikzstyle{pn}=[draw,thin,shape=circle,black,inner sep=0pt,minimum size=3pt]
\def\Pf{\operatorname{Pf}}
\def\e{\epsilon}

\def\tr{\operatorname{tr}}
\def\coleq{\mathrel{\mathop:}=}

\newcommand\poly[1]{\mathbf #1}
\def\Ho{\mathcal{H}_{1}}
\def\HN{\mathcal{H}^{B}_{\ell}}
\def\HC{\mathcal{H}_{\ell}}
\begin{document}
\begin{abstract}
We provide explicit formulae for highest-weight to highest-weight
correlation functions of perfect vertex operators
of $U_q(\widehat{\mathfrak{sl}(2)})$ in arbitrary integer level $\ell$.
They are given in terms of certain Macdonald polynomials.
We apply this construction to the computation of the ground
state of higher spin vertex models, spin chains (spin $\ell/2$
XXZ) or loop models in the root of unity case $q=-e^{-i\pi/(\ell+2)}$.
\end{abstract}

\maketitle

\section{Introduction}
In the seminal paper \cite{FR-qKZ}, Frenkel and Reshetikhin showed
that correlation functions of $q$-deformed vertex operators (VOs) associated
to quantized affine algebras satisfy a set of holonomic $q$-difference
equations which are a $q$-deformation of the Knizhnik--Zamolodchikov equation
\cite{KZ} ($q$KZ). The analysis is performed for generic value of the level
$\ell$ of the affine algebra, but also works for positive integer level
provided the obvious modifications are made. The most important one is that
the paths in the Weyl chamber describing
the various conformal blocks (correlation functions of vertex operators)
have to be restricted to the ``Weyl alcove''.
In all that follows
we focus on the simplest algebra, $U_q(\widehat{\mathfrak{sl}(2)})$.
In this case, the integrable irreducible highest weight modules
are characterized by a spin $s\in\frac{1}{2}\mathbb{Z}_+$
(defined e.g.\ as the spin of the $U_q(\mathfrak{sl}(2))$ 
representation
of the top degree part, where $U_q(\mathfrak{sl}(2))$ is the horizontal
subalgebra of $U_q(\widehat{\mathfrak{sl}(2)})$
), which in level $\ell$ only exist when $s\le\ell/2$.
Each vertex operator corresponds to a step
of the path and is itself characterized by a spin $j\in\{1/2,\ldots,\ell/2\}$.
The present paper is entirely dedicated to the case of so-called
perfect vertex operators, that is,
$j$ has its maximal value $j=\ell/2$, for which there is exactly one possible step,
namely, $s\to\ell/2-s$. This implies that various simplifications occur,
and we expect this unique conformal block to be particularly simple;
and indeed we provide an explicit formula for them.

The reason to revisit this somewhat old subject is the observation made in \cite{artic34} that certain solutions of $q$KZ also provide the ground state entries
of integrable models at special values of the deformation parameter $q$.
In fact, this idea can be traced back to \cite{Resh-OBA}, where
Reshetikhin obtained solutions of $q$KZ using an ``off-shell Bethe
Ansatz'' and observed that Bethe equations appeared in the semi-classical
approximation, and then pursued in
\cite{RV-KZ}, where  the Knizhnik--Zamolodchikov equation in the quasi-classical
limit was found to produce eigenvectors of the Gaudin model.
Here, we are only concerned with 
one especially simple solution of $q$KZ -- with appropriate normalization,
a {\em polynomial}\/ solution, which produces
one eigenvector
of the higher spin $U_q(\widehat{\mathfrak{sl}(2)})$ (inhomogeneous, twisted)
transfer matrix.

Note that in earlier work \cite{artic34}
(see also \cite{artic31,artic42}), only level $1$ solutions were used, whose
explicit form was already known and can be found e.g.\ in \cite{JM-book}. 
In the present paper, we obtain a novel expression for the 
correlation functions of 
level $\ell$ perfect VOs, and show that once specialized,
they provide the ground state entries
of the integrable spin $\ell/2$ chain when $q=-e^{-i\pi/(\ell+2)}$.
We also establish the connection with the loop model of \cite{artic37} and thus
prove Conj.~1 in it, which is concerned with the degree
of the ground state entries as polynomials. 

Also note that a {\em different}\/ procedure to produce (arbitrary level)
solutions of the $q$KZ equation is to compute finite temperature
correlation functions of the type
$\tr(x^d \Phi\ldots\Phi)$, see e.g.~\cite{JM-book}. 
However
the resulting integral formulae are of a more complicated nature than those
we consider here. 
In particular, they only make sense for $|q|<1$
and it is not clear how to continue them to $|q|=1$. 
Here we restrict ourselves to the zero temperature
correlation functions of the type $\braC{0}\Phi\ldots\Phi\ketC{0}$,
though $\tr(x^d \Phi\ldots\Phi)$ should be computable along similar lines.

The plan of the paper is as follows: in section~\ref{sec:vo} we review
the construction of $U_q(\widehat{sl(2)})$ currents in terms
of bosons and parafermions and then proceed
to build perfect vertex operators. Section~\ref{sec:corr_Uqsl2}, 
the core of the paper, contains the derivation
of the correlation function of these vertex operators. Section~\ref{sec:applic} 
then reinterprets this correlation function as the 
eigenvector of an integrable transfer matrix in various contexts (spin
chain, loop model, supersymmetric lattice fermions). Finally,
various appendices provide additional technical details.

\section{Construction of \texorpdfstring{$U_q(\widehat{sl(2)})$}{Uq(sl(2))} perfect vertex operators}\label{sec:vo}
\subsection{The algebra \texorpdfstring{$U_q(\widehat{sl(2)})$}{Uq(sl(2))}}
Our reference for the quantized affine
algebra $U=U_q(\widehat{sl(2)})$ is the book
\cite{JM-book} as well as the paper \cite{IIJMNT}, whose conventions we follow.
The Chevalley generators are denoted by $E_i, F_i, K_i$, $i=0,1$,
to which one must add the grading operator $d$: we choose the
homogeneous gradation, i.e., $E_0$ (resp.\ $F_0$) has degree $+1$ (resp.\ $-1$)
and all other generators have degree $0$. In particular there is a horizontal
subalgebra $U_1\subset U$ which is generated by $E_1,F_1,K_1$.

We also use in what follows Drinfeld's realization
of $U$ in terms of currents \cite{Drinf-cur}.

Until specified otherwise we are in the regime $|q|<1$.

We shall consider the spin $\ell/2$ evaluation representation $\rho_z$ acting on
$V_z\cong \mathbb{C}^{\ell+1}$ with standard basis
$\ketE{b}$, $b=0,\ldots,\ell$. It is defined by:
\begin{equation}\label{eq:evalrep}
\begin{aligned}
\rho_z(F_1)\ketE{b}&=[b] \ketE{b-1}
&
\rho_z(F_0)\ketE{b}&=[\ell-b] z^{-1} q^{\ell+2} \ketE{b+1}
\\
\rho_z(E_1)\ketE{b}&=[\ell-b] \ketE{b+1}
&
\rho_z(E_0)\ketE{b}&=[b] z q^{-\ell-2} \ketE{b-1}
\\
\rho_z(K_1)\ketE{b}&=q^{2b-\ell} \ketE{b}
&
\rho_z(K_0)\ketE{b}&=q^{\ell-2b} \ketE{b}
\end{aligned}
\end{equation}
where $[n]:=(q^n-q^{-n})/(q-q^{-1})$.
In order to make $d$ act, 
one must consider $z$ a formal variable,
in which case $d=z\frac{d}{dz}+\Delta$, where $\Delta$ is a constant
to be fixed below.
 When restricted to $E_1, F_1, K_1$,
it is an ordinary irreducible representation of $U_1=U_q(\mathfrak{sl}(2))$
of spin $\ell/2$, and we simply denote the underlying space $V$.

We shall also need the $R$-matrix $\bar R(z_1/z_2)$ acting on a pair
of such representations $V_{z_1}\otimes V_{z_2}$.
The bar denotes the fact that we choose
for now a particular normalization of the $R$-matrix where it is the
identity on the highest weight vectors. Other normalizations will be considered
below. Equivalently, consider
$\check{\bar R}(z)=\mathcal{P}\bar R(z)$ where $\mathcal{P}$ permutes factors
of the tensor product. 
We then have the following formula:
\begin{equation}\label{eq:defR}
\check{\bar R}(z)=\sum_{j=0}^\ell \prod_{r=j+1}^\ell \frac{z-q^{2r}}{1-q^{2r}z}
P_{j}
\end{equation}
where $P_{j}$ is the projector onto the spin $j$ irreducible subrepresentation
of $V\otimes V$ w.r.t.\ the $U_1$-action.

Finally, we shall denote collectively by $\HC$ any level $\ell$
representation space of $U$, with $\sigma: U\to \mathcal{L}(\HC)$.

\subsection{Level \texorpdfstring{$\ell$ $U_q(\widehat{sl(2)})$}{l Uq(sl(2))} currents}
The results which are summarized here are based on the work of
 Ding and Feigin \cite{DF-para} (see also \cite{BV-para}). 
For the sake of completeness and because some details do not appear in this article we give a detailed introduction to the subject in appendix~\ref{app:parafermions}.

Consider the bosonic operators:
\begin{align}
 [\bm a_n,\bm a_m] &= \delta_{m,-n} \frac{[2n][\ell n]}{n} & [\bm a_n,\bm \beta] &= \delta_{n,0} & n,m & \in \mathbb{Z}
\end{align}
with grading
\[
[d,\bm a_n]=n \bm a_n\qquad
[d,\bm\beta]=-\frac{\bm a_0}{2\ell}
\]
They act on the direct sum of copies of the bosonic Fock space
$\HN = \bigoplus_{h\in\mathbb{Z}}
\left\langle\prod_{m>0} \bm a_{-m}^{k_m}\ketN{h}\right\rangle$
which differ by the eigenvalue of $\bm a_0$: $\bm a_0\ketN{h}=h\ketN{h}$.
$\bm \beta$ only appears via its exponential, which is such that
$e^{\pm\bm\beta}\ketN{h}=\ketN{h\pm 1}$. The grading is given by
$d\ketN{h}=-\frac{1}{4\ell}h^2\ketN{h}$.
When $\ell=1$, this is a realization of $U$, and therefore it is denoted by $\Ho$.

We build the following current operators:
\begin{equation}
\begin{aligned}
 \bm{e}(z) &= \xi^+(z) e^{2\bm{\beta}} z^{\nicefrac{\bm{a}_0}{\ell}} e^{\sum_{n>0} q^{-\frac{n\ell}{2}} \frac{\bm{a}_{-n}}{[n \ell]} z^n} e^{-\sum_{n>0} q^{-\frac{n\ell}{2}} \frac{\bm{a}_n}{[n \ell]} z^{-n}}\\
 \bm{f}(z) &= \xi^- (z) e^{-2\bm{\beta}} z^{-\nicefrac{\bm{a}_0}{\ell}} e^{-\sum_{n>0} q^{\frac{n\ell}{2}} \frac{\bm{a}_{-n}}{[n \ell]} z^n} e^{\sum_{n>0} q^{\frac{n\ell}{2}} \frac{\bm{a}_n}{[n \ell]} z^{-n}}\\
 \bm k^{\pm} (z) &= q^{\pm \bm{a}_0} e^{\pm (q-q^{-1}) \sum_n \bm{a}_{\pm n} z^{\mp n}}
\end{aligned}
\end{equation}
where the $\xi^\pm(z)$ are parafermions, see appendix~\ref{app:parafermions}.
A full list of relations they satisfy is given in \eqref{eq:pf_relations},
as well as the relation with Chevalley generators. Here we list a few:
\begin{equation}
\begin{aligned}
 \bm k^+(z) \bm f(w) &= \frac{z-q^{2+\frac{\ell}{2}} w}{q^2 z - q^{\frac{\ell}{2}}w} \bm f(w) \bm k^+ (z)\\
 \bm e(z) \bm e(w) \left(z - q^2 w\right) &= \bm e(w) \bm e(z) \left(q^2 z - w\right) \\
 \bm f(z) \bm f(w) \left(q^2 z - w\right) &= \bm f(w) \bm f(z) \left(z - q^2 w\right) \\
 (q-q^{-1}) zw[\bm e(z),\bm f(w)] &= \delta \left( q^{-\ell} \frac{z}{w} \right) \bm k^+ (q^{-\nicefrac{\ell}{2}} z) - \delta \left( q^{\ell} \frac{z}{w} \right) \bm k^- (q^{\nicefrac{\ell}{2}} z)
\end{aligned}
\end{equation}
where $ \delta(z) = \sum_{i \in \mathbb{Z}} z^i $.

These operators act on the tensor product of $\ell$ copies of the Fock space $\Ho$, wich is a level $\ell$ representation of $U$.
A state $\ketN{k_1} \otimes \ldots \otimes \ketN{k_{\ell}}$ is denoted by $\ketC{k_1,\ldots,k_{\ell}}$.
As we shall see in Section~\ref{sec:corr_Uqsl2}, when 
$k_i = 0$ or $1$ for all $i$,
$\ketC{k_1,\ldots,k_{\ell}}$ is a highest weight vector, with highest weight
given by the sum $k=\sum_i k_i$; for fixed $k=\sum_i k_i$, these states
are interchangeable and we denote any of them by $\ketC{k}$.

Remark: for an alternative construction of these vertex operators,
see \cite{BW-VO}.

\subsection{Intertwining relations}

Let $\Phi(z)$ be a type I vertex operator,
i.e., an intertwiner
\[
\Phi(z): \HC \to \HC\otimes V_z
\]
whose normalization will be fixed later.

The intertwining condition writes:
\begin{equation}
(\sigma\otimes\rho_z)(\Delta(x))\Phi(z)=\Phi(z)\sigma(x)
\end{equation}
for all $x\in U$,
or explicitly
for the Chevalley generators:
\begin{equation}
\begin{aligned}
(\sigma(E_i)\otimes 1+\sigma(K_i)\otimes \rho_z(E_i))\Phi(z)
&=\Phi(z)\sigma(E_i)
\\
(\sigma(F_i)\otimes \rho_z(K_i^{-1})+1\otimes \rho_z(F_i))\Phi(z)
&=\Phi(z)\sigma(F_i)
\\
(\sigma(K_i)\otimes \rho_z(K_i))\Phi(z)&=\Phi(z)\sigma(K_i)
\end{aligned}
\end{equation}
plus the intertwining condition for $d$ which will be discussed separately
below.

In what follows application of $\sigma$ will be implicit.
Using the form of $\rho_z$, cf~\eqref{eq:evalrep},
we can expand in components: $\Phi(z)=\sum_{b=0}^\ell \Phi_b (z)\ketE{b}$
with
\begin{equation}
\begin{aligned}\label{eq:Phi_SL2}
\Phi_b(z)E_0-E_0\Phi_b(z)
&= [b+1] z q^{-\ell-2} K_0 \Phi_{b+1}(z)
\\
\Phi_b(z)F_1-{q^{\ell-2b}}F_1\Phi_b(z)
&=[b+1] \Phi_{b+1}(z)
\\
\Phi_b(z)E_1-E_1\Phi_b(z)
&= [\ell-b+1] K_1 \Phi_{b-1}(z)
\\
\Phi_b(z)F_0-{q^{2b-\ell}}F_0\Phi_b(z)
&=[\ell-b+1] z^{-1} q^{\ell+2} \Phi_{b-1}(z)
\\
q^{2b-\ell} K_1\Phi_b(z)&=\Phi_b K_1
\end{aligned}
\end{equation}
Since $\HC$ is a level $\ell$ representation, we have $K_0K_1=q^{\ell}$.

We are mostly interested in operators $F_1$ and $E_0$; expanding
the lowering current $\bm f(z)$ in modes: 
$\bm f(z)=\sum_{n\in\mathbb{Z}} f_n z^{-n-1}$, we 
have $f_0=F_1$, $f_1=E_0 K_1$, so that 
we rewrite the two first equations in~\eqref{eq:Phi_SL2}:
\begin{equation}
\begin{aligned}
\Phi_{b}(z)f_0-{q^{\ell-2b}}f_0\Phi_{b}(z)
&= [b+1] \Phi_{b+1}(z)
\\
\Phi_b(z)f_1-q^{2b-\ell}f_1\Phi_{b}(z)&=q^{2b-\ell}[b+1] z \Phi_{b+1}(z)
\end{aligned}
\end{equation}
This means starting from $\Phi_0(z)$, one can build the $\Phi_b$
by iterated $q$-commutator.
We obtain:
\begin{equation}\label{eq:vop_b}
\begin{aligned}
\Phi_b(z)&=
\sum_{k=0}^b 
\frac{(-1)^k q^{k(\ell-b+1)}}{[k]![b-k]!} 
f_0^k \Phi_0(z)f_0^{b-k}
\\
&=z^{-b}\sum_{k=0}
\frac{(-1)^k q^{(b-k)(\ell-b+1)}}{[k]![b-k]!} 
f_1^k \Phi_0(z)f_1^{b-k}
\end{aligned}
\end{equation}
where $[n]!=[n][n-1]\ldots[1]$.

Finally, we write $f_0=\oint \bm f(w) dw$, $f_1=\oint \bm f(w) w dw$, so that:
\begin{equation}\label{eq:vop_c}
\begin{aligned}
\Phi_b(z)&=
\oint \ldots \oint \prod_{i=1}^b \frac{dw_i}{2\pi i}  
\sum_{k=0}^b 
\frac{(-1)^k q^{k(\ell-b+1)}}{[k]![b-k]!} 
\bm f(w_1) \ldots \bm f(w_k) \Phi_0 (z)\bm f (w_{k+1})\ldots \bm f(w_b)
\\
&=
\oint \ldots \oint \prod_{i=1}^b \frac{w_i dw_i}{2 z^b \pi i}  
\sum_{k=0}^b 
\frac{(-1)^k q^{(b-k)(\ell-b+1)}}{[k]![b-k]!} 
\bm f(w_1) \ldots \bm f(w_k) \Phi_0 (z)\bm f(w_{k+1})\ldots \bm f(w_b)
\end{aligned}
\end{equation}
The summation over $k$ will be performed in section~\ref{sec:vop_b}.

\subsection{Explicit formulae for the perfect vertex operators}
So far we have not used the fact that twice the spin of the VO is equal
to the level.
For so-called perfect VOs, that is for rank 1 case precisely
the VOs of maximal
spin $\ell/2$, we also have the following form of the highest weight entry
in terms of the bosonic field:
\begin{equation}
 \Phi_0 (z) = e^{\ell\bm \beta} (-z)^{\nicefrac{\bm a_0}{2}} e^{\sum_n q^{\frac{n \ell}{2}} \frac{\bm a_{-n}}{[2n]} z^n} e^{- \sum_n q^{\frac{n \ell}{2}} \frac{\bm a_{n}}{[2n]} z^{-n}}
\end{equation}
In order to check the consistency of this Ansatz, we shall require
certain identities.

\subsubsection{Commutation relations}
In order to compute the commutation relations between the operators $\bm e (w)$, $\bm f (w)$ and $\Phi_0 (z)$, it is convenient to
define the normal ordering $\normord{\ldots}$, which pushes the negative modes to the left, that is:
\begin{align*}
 \left. \begin{array}{c} \normord{\bm a_{-m} \bm a_n}\\\normord{\bm a_n \bm a_{-m}} \end{array} \right\} &= \bm a_{-m} \bm a_n & \left. \begin{array}{c}\normord{\bm a_0 \bm \beta} \\ \normord{\bm \beta \bm a_0} \end{array} \right\}&= \bm \beta \bm a_0
\end{align*}
where $n$ and $m$ are positive integers.
The normal ordering has the following important properties:
\begin{lemma}\label{prop:no_several}
 Let $\mathcal{O}$ be a product of $N$ operators $\mathcal{O} = \mathcal{O}_1 \ldots \mathcal{O}_N$.
 If $\mathcal{O}_i \mathcal{O}_j = g_{i,j} \normord{\mathcal{O}_i \mathcal{O}_j}$, then $\mathcal{O} = \prod_{i<j} g_{i,j} \normord{\mathcal{O}}$.
\end{lemma}

\begin{lemma}\label{prop:no_scalar}
 Let $\mathcal{O}$ be an operator such that $\normord{\mathcal{O} }= e^{g \bm \beta} f(\bm a_0) e^{\sum_{n>0} d_{n} \bm a_{-n}} e^{\sum_{n>0} c_n \bm a_n}$, where $g$, $c_n$ and $d_n$ are scalar constants.
 Then
\[
  \braN{k} \normord{\mathcal{O}} \ketN{k} = f(k)
\]
if $g=0$ and it vanishes otherwise.

 In the more general context where $\normord{\mathcal{O}} = \bigotimes_i e^{g_i \alpha} f_i(a_0) e^{\sum_{n>0} d_{i,n} a_{-n}} e^{\sum_{n>0} c_{i,n} a_n}$ this property is conserved:
\[
 \braC{k_1,\ldots,k_{\ell}} \normord{\mathcal{O}} \ketC{k_1,\ldots,k_{\ell}} = \prod_i f_i (k_i) 
\]
if $g_i = 0$ for all $i$ and it vanishes otherwise.
\end{lemma}

For convenience, we split the currents in bosonic and parafermionic components,
that is, $\bm e (w) = \xi^+ (w) \mathfrak{e} (w)$ and $\bm f (w) = \xi^- (w) \mathfrak{f} (w)$.

\begin{proposition}\label{prop:normal_order}
 The relation between the products and the normal ordered products is:
\begin{align*}
 \Phi_0 (z_1) \Phi_0 (z_2) &= (-z_1)^{\frac{\ell}{2}} \frac{\poc{q^2 \frac{z_2}{z_1}}}{\poc{q^{2\ell+2} \frac{z_2}{z_1}}} \normord{\Phi_0 (z_1) \Phi_0 (z_2)}\\
 \mathfrak{f} (w) \Phi_0 (z) &= \frac{1}{w-q^{\ell} z} \normord{\mathfrak{f} (w) \Phi_0 (z)}\\
 \Phi_0 (z) \mathfrak{f} (w) &= \frac{1}{q^{\ell} w -z} \normord{\mathfrak{f} (w) \Phi_0 (z)}\\
 \mathfrak{e} (w) \Phi_0 (z) &= (w-z) \normord{\mathfrak{e} (w) \Phi_0 (z)}\\
 \Phi_0 (z) \mathfrak{e} (w) &= (w -z) \normord{\mathfrak{e} (w) \Phi_0 (z)}
\end{align*} 
\end{proposition}

\begin{remark}\label{rmk:normal_order}
 In the previous computations, we perform infinite sums which 
are convergent under the following conditions:
 \begin{align*}
| q^2 z_2| &< |z_1| &
 |q^{\ell} z| &< |w| &
 |q^{\ell} w| &< |z| &
 |z| &< |w| &
 |w| &< |z| 
 \end{align*}
for each expression in proposition~\ref{prop:normal_order}, respectively.
\end{remark}

From these expressions we find immediately:
\begin{equation}
\begin{aligned}
\mathfrak{e} (w) \Phi_0(z)&= \mathfrak e(w)\Phi_0(z)\\
 \mathfrak{f}(w) \Phi_0 (z) &= \frac{q^{\ell}w-z}{w-q^{\ell}z} \Phi_0 (z) \mathfrak{f} (w)
\end{aligned}
\end{equation}
As proven in appendix~\ref{app:parafermions}, the parafermions do not interact with the bosons. 
Therefore, in these formulae we can replace $\mathfrak{f}(w)$ with $\bm f(w)$.

From the first identity we conclude that $[e_k,\Phi_0(z)]=0$.
For $k=0,1$, this coincides with 
the intertwining equations 
in~\eqref{eq:Phi_SL2} for $E_1$ and $F_0$ at $b=0$.
We of course also have $K_1\Phi_0=q^\ell \Phi_0K_1$,
using $K_1=q^{\bm a_0}$.
Finally, we have:
\begin{equation}\label{eq:dimPhi}
x^{d}\Phi_0(z)x^{-d}=x^{-\ell/4}\Phi_0(x^{-1}z)
\end{equation}
which coincides with the intertwining condition for $d$ on condition
that one set $\Delta=\ell/4$.
\footnote{Note that in Conformal Field Theory
(corresponding to $q\to1$), this would be interpreted
as a consequence of the conformal dimension $\ell/4$ of $\Phi_0(z)$.
The $q$-deformation breaks the conformal symmetry by marking the two
points $0$ and $\infty$, leaving only the residual scaling symmetry
$z\to xz$.}

\subsubsection{Closed expression for the $\Phi_b (w)$}\label{sec:vop_b}
Equation~\eqref{eq:vop_b} provides us with two expressions for $\Phi_b (z)$.
Performing the summation results in:
\begin{lemma}\label{lemma:Phi_b}
The following expression holds:
\begin{align*}
\Phi_b(z)&=\alpha_b
\oint \ldots \oint \prod_{i=1}^b \frac{w_i dw_i}{2\pi i}  
\frac{\Phi_0 (z)\bm f (w_1)\ldots \bm f(w_b)}{\prod_i (w_i-q^{\ell} z)}
\end{align*}
where $\alpha_b$ is a constant given by:
\[
 \alpha_b = \frac{\prod_{i=1}^{b}(1-q^{2\ell +2 -2i}) }{[b]!}= (-1)^b (q-q^{-1})^b q^{b(2\ell-b+1)/2} \left[{\ell\atop b}\right]
\]
with the notation $\left[{\ell\atop b}\right]=\frac{[\ell]!}{[b]![\ell - b]!}$.
\end{lemma}

\begin{proof}
We shall use the first expression of $\Phi_b(z)$ in \eqref{eq:vop_c}; the second expression would lead to the same result.
We have a sum over several terms.
We commute all $\bm f$ to the right, so that each term becomes:
\begin{multline*}
 \bm f (w_1) \ldots \bm f (w_k) \Phi_0 (z) \bm f (w_{k+1}) \ldots \bm f (w_b) =
  \frac{\prod_{j\leq k} (q^{\ell}w_j -z )}{\prod_{j\leq k}(w_j-q^{\ell} z)}
\Phi_0 (z) \bm f (w_1) \ldots \bm f (w_b)\\
= \frac{\prod_{j\leq k} (q^{\ell} w_j-z)\prod_{j>k}(w_j-q^{\ell} z)}{\prod_j (w_j-q^{\ell} z)}  
\Phi_0 (z) \bm f (w_1) \ldots \bm f (w_b)
\end{multline*}

Our purpose is to compute the expectation value of an operator, which contains the operator $\Phi_0(z) \bm f(w_1) \ldots \bm f(w_b)$.
By lemmas~\ref{prop:no_several} and~\ref{prop:no_scalar}, this will make appear one term per each pair of operators.
From $\Phi_0 (z) \bm f(w_i)$ we get a pole $\frac{1}{q^{\ell} w_i-z}$.
The contribution of the pairs $\bm f(w_i) \bm f(w_j)$ is more complicated and we need to consider the decomposition of parafermions in modes $\bm f^k (w_i)$, which we explain in appendix~\ref{app:parafermions}. 
Here we only need a couple of properties.
The relations~\eqref{eq:pf_relations}, which define $U_q(\widehat{\mathfrak{sl}(2)})$, imply that
\[
\bm f (w_1) \ldots \bm f (w_b) \prod_{i<j} \frac{q^2 w_i - w_j}{w_i - w_j }
\]
is symmetric in the exchange of the $w_i$.
This expression decomposes as a sum of terms of the form
$\bm f^{i_1} (w_1) \ldots \bm f^{i_b} (w_b) \prod_{i<j} (q^2 w_i - w_j)$,
which each satisfy by equation~\eqref{eq:pf_f_relations}
\[
\bm f^{i_1} (w_1) \ldots \bm f^{i_b} (w_b) \prod_{i<j} (q^2 w_i - w_j)
 = P(w_1, \ldots, w_b) \normord{\bm f^{i_1} (w_1) \ldots \bm f^{i_b} (w_b)}
\] 
where $P$ is some multivariate polynomial.
We conclude from the definition of normal ordering that
$\bm f (w_1) \ldots \bm f(w_b) \prod_{i<j} \frac{q^2 w_i - w_j}{w_i-w_j}$ is a symmetric function of the $\{w_1,\ldots,w_b\}$ without poles of the form $(w_i - \zeta w_j)^{-1}$.

In summary, if we write 
$
f^i_b(z;w_1,\ldots,w_b) = \prod_{j\leq i} (q^{\ell} w_j-z)\prod_{j>i}(w_j-q^{\ell}z)$
(which is linear in each variable $w_j$),
then the quantity that we want to compute is:
\begin{align*}
\Phi_b(z)&=
\oint \ldots \oint \prod_{i=1}^b \frac{dw_i}{2\pi i}  
\sum_{k=0}^b 
\frac{(-1)^k q^{k(\ell-b+1)}}{[k]![b-k]!} 
f_b^k (z;w_1,\ldots,w_b) 
\prod_{i<j} \frac{w_i-w_j}{q^2 w_i-w_j}\mathcal{S}(z;w_1,\ldots,w_b)
\end{align*} 
where $S(z;w_1,\ldots,w_b)$ is some symmetric function of the $\{w_1,\ldots,w_b\}$ without poles of the form $(w_i - \zeta w_j)^{-1}$.

We now want to prove that 
\[
\sum_{k=0}^b 
\frac{(-1)^k q^{k(\ell-b+1)}}{[k]![b-k]!} 
f_b^k (z;w_1,\ldots,w_b) 
= \alpha_b \prod_{i=1}^b w_i
\]
up to some terms which vanish when integrated, i.e., which are divisible by $(q^2 w_i- w_{i+1})$. They produce a zero contribution 
because the integral becomes skew-symmetric.

This is true when $b=1$:
\[
f_1^0 (z;w_1) -q^{\ell} f_1^1 (z;w_1) 
= (w_1-q^{\ell}z) -q^{\ell} (q^{\ell}w_1-z) 
= (1-q^{2 \ell}) w_1
\]

Suppose that this is true for $b-1$, then:
\begin{align*}
\sum_{k=0}^b 
\frac{(-1)^k q^{k(\ell-b+1)}}{[k]![b-k]!} f_b^k 
&=  \sum_{k=0}^{b-1} \frac{(-1)^k q^{k(\ell-b+2)}}{[k]![b-k-1]![b]} \prod_{j\leq k} (q^{\ell} w_j-z)\prod_{j>k}(w_j-q^{\ell}z)\\
&\quad- q^{\ell+2-2b}\sum_{k=1}^b \frac{(-1)^{k-1} q^{(k-1)(\ell-b+2)}}{[k-1]![b-k]![b]} \prod_{j\leq k} (q^{\ell} w_j-z)\prod_{j>k}(w_j-q^{\ell}z)\\
&= \frac{\alpha_{b-1}}{[b]} (w_b - q^{\ell} z) \prod_{i=1}^{b-1} w_i - \frac{q^{\ell+2-2b}\alpha_{b-1}}{[b]} (q^{\ell}w_1-z) \prod_{i=2}^b w_i\\
&= \frac{\alpha_{b-1}}{[b]} (1-q^{2\ell+2-2b}) \prod_{i=1}^b w_i - z q^{\ell}\frac{\alpha_{b-1}}{[b]} (w_1 - q^{2-2b}w_b) \prod_{i=2}^{b-1} w_i
\end{align*}
by the induction hypothesis.
The last term can be rewritten:
\[
 (w_1 - q^{2-2b}) w_2 \ldots w_{b-1} = \sum_{i=1}^{b-1}q^{-2(i-1)}\left(\prod_{j<i} w_j\right) (w_i - q^{-2} w_{i+1})\left(\prod_{j>i+1} w_j \right) 
\]
which vanishes when integrated.

The constant $\alpha_b$ is given by this recursion.
\end{proof}

\subsection{The \texorpdfstring{$R$}{R} matrix}

Consider next the operator $\check R(z_1/z_2)\Phi(z_1)\Phi(z_2): \HC\to \HC\otimes V_{z_2}\otimes V_{z_1}$, with the normalization of
the $R$-matrix to be adjusted below.
$\check R(z_1/z_2)\Phi(z_1)\Phi(z_2)$ is an intertwiner by the defining properties of $R$ and $\Phi$.
We now use the fact that $\Phi$ is a perfect VO, i.e.,
that it is of (maximal) spin $\ell/2$ where $\ell$ is the level.
This means that each highest weight
irreducible representation (say of spin $s$) in $\HC$,
when tensored with $V_z$, is sent onto another {\em unique}\/ irreducible
representation (namely, of spin $\ell/2-s$), and the same when tensored
twice.
Such an intertwiner is therefore unique;
so that it should be proportional to $\Phi(z_2)\Phi(z_1)$. In order
to fix the proportionality constant, we compute
\[
r(z_1/z_2)
\Phi^{(0)}(z_1)\Phi^{(0)}(z_2)
=
\Phi^{(0)}(z_2)\Phi^{(0)}(z_1)
\]
with
\[
r(z)=z^{-\ell/2}\frac{\poc{q^2z}\poc{q^{2\ell+2}/z}}{\poc{q^2/z}\poc{q^{2\ell+2}z}}
\]
and now choose (as in \cite{IIJMNT})
\[
R(z)=r(z)
\bar R(z)
\]
where $\bar R(z)$ is the identity on the tensor product of highest weight
vectors, in such a way that
\begin{equation}\label{eq:exchPhi}
\check R(z_1/z_2)\Phi(z_1)\Phi(z_2)=\Phi(z_2)\Phi(z_1)
\end{equation}


\section{Correlation functions of \texorpdfstring{$U_q(\widehat{sl(2)})$}{Uq(sl(2))} perfect vertex operators}\label{sec:corr_Uqsl2}
In this section we explain how to compute the correlation function
in even size $L=2n$:
\[
 \Psi^{(k)}_{b_1,\ldots,b_{2n}} (z_1,\ldots,z_{2n}) 
= \braC{k} \Phi_{b_1} (z_1) \ldots \Phi_{b_{2n}} (z_{2n}) \ketC{k}
\]
$\Psi^{(k)}_{b_1,\ldots,b_{2n}} (z_1,\ldots,z_{2n})$ is nonzero only when the neutrality condition $\sum_{i=1}^{2n} b_i = n \ell$ it satisfied.
The final result will be a multiple contour integral. 

\subsection{Integral formulae for correlation functions}
By lemma~\ref{lemma:Phi_b}, the product of vertex operators $\Phi_{b_1} (z_1) \ldots \Phi_{b_{2n}}(z_{2n})$ can be expressed as a multiple integral containing $2n$ operators $\Phi_0 (z_i)$ and $\ell n$ operators $\bm f(w_j)$.
Then we can use lemmas~\ref{prop:no_several} and~\ref{prop:no_scalar} to compute the correlation function.

\subsubsection{Correlation functions of $q$-parafermions}\label{sec:corrpara}
Take a product of $\ell n$ currents, multiplied as before
by some appropriate prefactors, and do a mode decomposition: 
\begin{equation}\label{eq:F_decomposition}
 \mathcal{F} =\bm{f}(w_1) \ldots \bm{f}(w_{\ell n}) \prod_{i<j} \frac{w_j-q^2 w_i}{w_j-w_i}
 = \sum_{i_1,\ldots,i_{\ell n}} \bm{f}^{i_1}(w_1) \ldots \bm{f}^{i_{\ell n}}(w_{\ell n}) \prod_{i<j} \frac{w_j-q^2 w_i}{w_j-w_i}
\end{equation}
Equation~\eqref{eq:pf_f_relations} implies that
the correlation function $\braC{0} e^{2 \ell n\bm \beta} \mathcal{F} \ketC{0}$ is a rational function.
The poles which appear on~\eqref{eq:pf_f_relations} cancel with the product $\prod_{i<j} (w_j - q^2 w_i)$, thus the expression is of the form:
\[
 \braC{0} e^{2 \ell n \bm \beta} \mathcal F \ketC{0} = \frac{\mathcal{P}(w_1, \ldots , w_{\ell n})}{\prod_{i<j} (w_i - w_j)}
\]
where $\mathcal{P}$ is some polynomial. 

By equation~\eqref{eq:pf_relations}, $\mathcal F$ is symmetric and therefore $\mathcal P$ is antisymmetric. It follows that $\braC{0} e^{2 \ell n \bm \beta} \mathcal F \ketC{0}$ is a symmetric polynomial in the $w_i$.\footnote{It can be shown that the coefficients of the polynomial are Laurent polynomials in $q$.} 
Moreover, as a consequence of theorem~\ref{thm:pf-wheel}
for its parafermionic part, 
it satisfies a {\em wheel condition}.

This wheel condition can be defined in a general setting as follows.
Let $t_M$, $q_M$ be two scalars such that $q_M^{r-1}t_M^{k+1} =1$.
A symmetric polynomial 
is said to satisfy the $(r,k)$--wheel condition, if it vanishes whenever we set the first $k+1$ variables such that $\frac{z_{i+1}}{z_i} = t_M q_M^{s_i}$ for all $i\leq k$ and $\sum_i s_i \leq r-1$.
Here we find that $\braC{0} e^{2 \ell n\bm \beta} \mathcal{F} \ketC{0}$ 
satisfies the wheel condition
when we set $r=2$, $k=\ell$ and $t_M=q^{-2}$.

For each term on the decomposition~\eqref{eq:F_decomposition} we can compute the effect of normal ordering:
\[
 \bm f^{i_1} (w_1) \ldots \bm f^{i_{\ell n}} (w_{\ell n}) = P^{i_1,\ldots,i_{\ell n}}(w_1,\ldots,w_{\ell n}) \normord{\bm f^{i_1} (w_1) \ldots \bm f^{i_{\ell n}} (w_{\ell n})} 
\]

Now, we ignore everything but the terms in $e^{\alpha}$.
On the one hand, we have $e^{2 \ell n \bm \beta} = \bigotimes_j e^{n\alpha}$.
On the other hand, each parafermionic mode $\bm{f}^k$ is proportional to $\bigotimes_{j<k} 1 \otimes e^{-\alpha} \bigotimes_{j>k} 1$.
Then, by lemma~\ref{prop:no_scalar}, each mode $k$ should appear exacly $n$ times.

Each time we have a repeated mode (that is, $\epsilon_i=\epsilon_j$) we get a term of degree $2$, more precisely $q^{2\ell-2\epsilon_i} (w_i-q^2w_j)(q^2 w_i-w_j)$.
Otherwise, we get a term of degree zero.
Therefore, the degree of $\braC{0} e^{2\ell n \bm \beta} \mathcal{F} \ketC{0}$ is $\ell n(n-1)$.
The same analysis can be used to prove that $\left(\prod_i w_i\right) \braC{k}e^{2\ell n \bm \beta} \ketC{k}$ is a symmetric polynomial of degree $\ell n(n-1) + n(\ell-k)$.

We can do a more detailed analysis and obtain the dominant monomial.
Let $Y_{\ell,n}$ be the staircase Young diagram, corresponding to $n$ steps of $\ell \times 2$: $Y_{\ell,n} = (2(n-1),\ldots,2(n-1),2(n-2),\ldots,2(n-2),2,\ldots,2,0,\ldots,0)$.
For example:
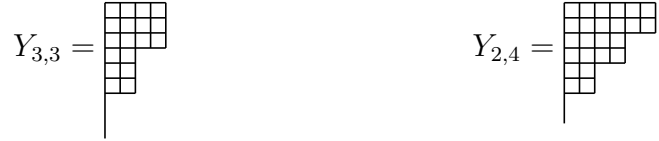
\begin{align*}
 Y_{3,3} &= 
\begin{tikzpicture}[scale=.2,baseline=-20]
 \draw[young] (0,0) -- (4,0);
 \draw[young] (0,-1) -- (4,-1);
 \draw[young] (0,-2) -- (4,-2);
 \draw[young] (0,-3) -- (4,-3);
 \draw[young] (0,-4) -- (2,-4);
 \draw[young] (0,-5) -- (2,-5);
 \draw[young] (0,-6) -- (2,-6);
 \draw[young] (0,0) -- (0,-9);
 \draw[young] (1,0) -- (1,-6);
 \draw[young] (2,0) -- (2,-6);
 \draw[young] (3,0) -- (3,-3);
 \draw[young] (4,0) -- (4,-3);
\end{tikzpicture}
&
 Y_{2,4} &= 
\begin{tikzpicture}[scale=.2,baseline=-20]
 \draw[young] (0,0) -- (6,0);
 \draw[young] (0,-1) -- (6,-1);
 \draw[young] (0,-2) -- (6,-2);
 \draw[young] (0,-3) -- (4,-3);
 \draw[young] (0,-4) -- (4,-4);
 \draw[young] (0,-5) -- (2,-5);
 \draw[young] (0,-6) -- (2,-6);
 \draw[young] (0,0) -- (0,-8);
 \draw[young] (1,0) -- (1,-6);
 \draw[young] (2,0) -- (2,-6);
 \draw[young] (3,0) -- (3,-4);
 \draw[young] (4,0) -- (4,-4);
 \draw[young] (5,0) -- (5,-2);
 \draw[young] (6,0) -- (6,-2);
\end{tikzpicture}
\end{align*}
More generally we define a modified staircase Young diagram $Y_{\ell,n}^{(k)}$. 
Which is obtained from $Y_{\ell,n}$ by adding $\ell-k$ boxes vertically  at each step.
For example:
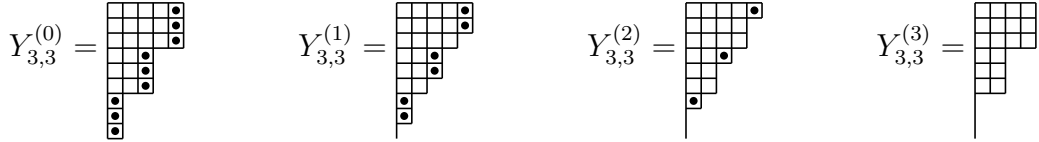
\begin{align*}
 Y^{(0)}_{3,3} &= 
\begin{tikzpicture}[scale=.2,baseline=-20]
 \draw[young] (0,0) -- (5,0);
 \draw[young] (0,-1) -- (5,-1);
 \draw[young] (0,-2) -- (5,-2);
 \draw[young] (0,-3) -- (5,-3);
 \draw[young] (0,-4) -- (3,-4);
 \draw[young] (0,-5) -- (3,-5);
 \draw[young] (0,-6) -- (3,-6);
 \draw[young] (0,-7) -- (1,-7);
 \draw[young] (0,-8) -- (1,-8);
 \draw[young] (0,-9) -- (1,-9);
 \draw[young] (0,0) -- (0,-9);
 \draw[young] (1,0) -- (1,-9);
 \draw[young] (2,0) -- (2,-6);
 \draw[young] (3,0) -- (3,-6);
 \draw[young] (4,0) -- (4,-3);
 \draw[young] (5,0) -- (5,-3);
 \node[dot] at (4.5,-0.5) {};
 \node[dot] at (4.5,-1.5) {};
 \node[dot] at (4.5,-2.5) {};
 \node[dot] at (2.5,-3.5) {};
 \node[dot] at (2.5,-4.5) {};
 \node[dot] at (2.5,-5.5) {};
 \node[dot] at (0.5,-6.5) {};
 \node[dot] at (0.5,-7.5) {};
 \node[dot] at (0.5,-8.5) {};
\end{tikzpicture}
&
 Y^{(1)}_{3,3} &= 
\begin{tikzpicture}[scale=.2,baseline=-20]
 \draw[young] (0,0) -- (5,0);
 \draw[young] (0,-1) -- (5,-1);
 \draw[young] (0,-2) -- (5,-2);
 \draw[young] (0,-3) -- (4,-3);
 \draw[young] (0,-4) -- (3,-4);
 \draw[young] (0,-5) -- (3,-5);
 \draw[young] (0,-6) -- (2,-6);
 \draw[young] (0,-7) -- (1,-7);
 \draw[young] (0,-8) -- (1,-8);
 \draw[young] (0,0) -- (0,-9);
 \draw[young] (1,0) -- (1,-8);
 \draw[young] (2,0) -- (2,-6);
 \draw[young] (3,0) -- (3,-5);
 \draw[young] (4,0) -- (4,-3);
 \draw[young] (5,0) -- (5,-2);
 \node[dot] at (4.5,-0.5) {};
 \node[dot] at (4.5,-1.5) {};
 \node[dot] at (2.5,-3.5) {};
 \node[dot] at (2.5,-4.5) {};
 \node[dot] at (0.5,-6.5) {};
 \node[dot] at (0.5,-7.5) {};
\end{tikzpicture}
&
 Y^{(2)}_{3,3} &= 
\begin{tikzpicture}[scale=.2,baseline=-20]
 \draw[young] (0,0) -- (5,0);
 \draw[young] (0,-1) -- (5,-1);
 \draw[young] (0,-2) -- (4,-2);
 \draw[young] (0,-3) -- (4,-3);
 \draw[young] (0,-4) -- (3,-4);
 \draw[young] (0,-5) -- (2,-5);
 \draw[young] (0,-6) -- (2,-6);
 \draw[young] (0,-7) -- (1,-7);
 \draw[young] (0,0) -- (0,-9);
 \draw[young] (1,0) -- (1,-7);
 \draw[young] (2,0) -- (2,-6);
 \draw[young] (3,0) -- (3,-4);
 \draw[young] (4,0) -- (4,-3);
 \draw[young] (5,0) -- (5,-1);
 \node[dot] at (4.5,-0.5) {};
 \node[dot] at (2.5,-3.5) {};
 \node[dot] at (0.5,-6.5) {};
\end{tikzpicture}
&
 Y^{(3)}_{3,3} &= 
\begin{tikzpicture}[scale=.2,baseline=-20]
 \draw[young] (0,0) -- (4,0);
 \draw[young] (0,-1) -- (4,-1);
 \draw[young] (0,-2) -- (4,-2);
 \draw[young] (0,-3) -- (4,-3);
 \draw[young] (0,-4) -- (2,-4);
 \draw[young] (0,-5) -- (2,-5);
 \draw[young] (0,-6) -- (2,-6);
 \draw[young] (0,0) -- (0,-9);
 \draw[young] (1,0) -- (1,-6);
 \draw[young] (2,0) -- (2,-6);
 \draw[young] (3,0) -- (3,-3);
 \draw[young] (4,0) -- (4,-3);
\end{tikzpicture}
\end{align*}
where the doted boxes correspond to the boxes added to $Y_{\ell,n}$.

Recall that a monomial corresponding to some partition $\lambda = \{\lambda_1,\ldots,\lambda_N\}$ is defined by $w^{\lambda} = \prod_i w_i^{\lambda_i}$.
\begin{proposition}\label{prop:leading}
 The leading monomial of $\left(\prod_i w_i\right)\braC{k} e^{2\ell n \bm \beta} \mathcal{F} \ketC{k}$ is:
\[
 q^{\ell^2 n(n-1)+n\binom{\ell}{2}} \left([k]![\ell-k]!\right)^n w^{Y_{\ell,n}^{(k)}}
\]
In the case $k=0$, the expression simplifies:
\[
 \braC{0} e^{2\ell n \bm \beta} \mathcal{F}\ketC{0} = q^{\ell^2 n(n-1) + n\binom{\ell}{2}} \left([\ell]!\right)^n w^{Y_{\ell,n}} + \text{ lower terms.}
\]
\end{proposition}
Sketch of proof: multiply $\left(\prod_i w_i \right) \braC{k} e^{2\ell n\bm \beta} \mathcal{F} \ketC{k}$ by the Vandermonde determinant, and try to maximize the degree of $w_1$, then try to maximize the degree of $w_2$ and so on.

Call a Young diagram $(r,k)$-admissible if $\lambda_i-\lambda_{i+k} \geq r$ for all $i$.
Then, it is clear that a staircase $(r,k)$-admissible Young diagram has the minimum number of boxes possible.
If we set $r=2$, $k=\ell$ and we restrict ourselves to diagrams of lenght $2n$ the only diagram with $\ell n(n-1)$ boxes is exacly $Y_{\ell,n}$.
For the case of modified staircases, the situation is different: there are several $(2,\ell)$-admissible Young diagrams with $\ell n(n-1) + (\ell-k) n$ boxes, the smallest of them being $Y_{\ell,n}^{(k)}$.

Macdonald polynomials~\cite{Macdonald} form a basis for the symmetric polynomials depending on two extra variables which we call $q_M$ and $t_M$.
The only result related to Macdonald polynomials that we need
for our purposes is the following theorem:
\begin{theorem}[Feigin, Jimbo, Miwa and Mukhin~\cite{FJMM-Macdonald}]\label{thm:FJMM}
 Let $q_M$, $t_M$ be such that $q_M^{r-1}t_M^{k+1}=1$.
 Let $\mathcal{V}$ be the vector space of symmetric polyynomials on $N$ variables satisfying the $(r,k)$--wheel condition.
 And let $\mathcal{M}$ be the vector space spanned by the Macdonald polynomials given by $(r,k)$-admissible Young diagrams and the Macdonald parameters $q_M$, $t_M$.

 Then $\mathcal{V}=\mathcal{M}$.
\end{theorem}

As a direct consequence we can compute the correlation function:
\begin{theorem}
\[
 \braC{0} e^{2 \ell n \bm \beta} \bm{f}(w_1) \ldots \bm{f}(w_{\ell n}) \ketC{0} 
\prod_{i<j} \frac{w_j-q^2 w_i}{w_j- w_i}
= \gamma_{\ell,n}  P_{Y_{\ell,_n}} (w_1,\ldots,w_{\ell n}) 
\]
where $P_{Y_{\ell,n}}$ is the Macdonald polynomial with parameters $t_M = q^{-2}$ and $q_{M} = t_M^{-(\ell+1)}$.
\end{theorem}
The proportionality factor is the leading coefficient, i.e. 
\[
 \gamma_{\ell,n} = q^{\ell^2 n(n-1) + n\binom{\ell}{2}} \left([\ell]!\right)^n
\]

\begin{proof}
 We have proven above that the correlation function in the theorem is a homogeneous symmetric polynomial that satisfies the wheel condition, more precisely the $(2,\ell)$-wheel condition.
 Therefore it lives in a vector subspace spanned by Macdonald polynomials associated with $(2,\ell)$-admissible Young diagrams. 
  
 The polynomial depends in $2n$ variables and has a combined degree of $\ell n (n-1)$.
 The diagram $Y_{\ell,n}$ is the smallest one that satisfies the wheel condition (with $2n$ parts) and it has exacly $\ell n (n-1)$.
 Therefore, it corresponds to the only Macdonald polynomial with such conditions.
\end{proof}

As $\braC{k} e^{2\ell n \beta} \mathcal{F} \ketC{k}$ also satisfies the wheel condition, we can express it as well as a sum over Macdonald polynomials:
$
 \left(\prod_i w_i\right) \braC{k} e^{2\ell n \beta} \mathcal{F} \ketC{k} = \sum_Y c_{k,Y} P_Y(w_1, \ldots, w_{\ell n})
$
where the sum runs over all $(2,\ell)$-admissible Young diagrams with $\ell n (n-1) + n (\ell -k)$ boxes.
By proposition~\ref{prop:leading}, we know that the leading monomial is $w^{Y_{\ell,n}^{(k)}}$ and therefore only term appears in this sum:
\[
 \left(
\prod_{i=1}^{\ell n} w_i
\right)
\braC{k} e^{2\ell n \beta} \mathcal{F} \ketC{k} = \gamma_{\ell,n}^{(k)} P_{Y_{\ell,n}^{(k)}}(w_1, \ldots, w_{\ell n})
\]
with coefficient
\[
 \gamma_{\ell,n}^{(k)} = q^{\ell^2 n(n-1) + n\binom{\ell}{2}} \left([k]![\ell-k]!\right)^n
\]
Notice that this implies that $P_{Y_{\ell,n}^{(0)}} = \left(\prod_i w_i \right) P_{Y_{\ell,n}}$.
\subsubsection{Derivation of integral formulae}
We have computed all the ingredients to compute $\Psi^{(k)}_{b_1,\ldots,b_{2n}}$.
In order to express the result, we use the following reparameterization
of the index sequence $\{b_1,\ldots,b_{2n}\}$.
Let $\bm \e=\{\e_1,\ldots,\e_{\ell n}\}$ be a non-decreasing sequence, such that $1\leq \e_i \leq 2n$, and such that, the same value can not be repeated more that $\ell$ times.
Then, there is a bijection between the set of neutral spin sequences $\{b_1,\ldots,b_{2n}\}$ and the non-decreasing sequences $\bm \e$, given by: $b_j = \#\{\e_i \ \text{such that }\e_i=j\}$.
To each $\e_i$ equal to $j$ corresponds a current $\bm f(w_i)$ that is used to raise the spin of $\Phi_0 (z_j)$.

We then compute:
\begin{itemize}
 \item From the product $\Phi_0 (z_1) \ldots \Phi_0 (z_{2n})$ we get: 
\[
\Phi_0 (z_1) \ldots \Phi_0 (z_{2n}) =
 \prod_{i<j} (-z_i)^{\frac{\ell}{2}} \frac{\left(q^2 \frac{z_j}{z_i};q^4\right)_{\infty}}{\left(q^{2\ell+2} \frac{z_j}{z_i};q^4\right)_{\infty}} \normord{\Phi_0 (z_1) \ldots \Phi_0 (z_{2n})}
\]

 \item When $\e_i < j$, we have pairs of the form $\bm f (w_i) \Phi_0 (z_j)$, and then:
\[
 \bm f (w_i) \Phi_0 (z_j) =  \frac{1}{w_i-q^{\ell}z_j} \normord{\bm f(w_i) \Phi_0 (z_j)}
\]
 \item When $\e_i \geq j$, we have pairs of the form $\Phi^0 (z_j) \bm f (w_i)$, and then:
\[
 \Phi_0 (z_j) \bm f (w_i) =  \frac{1}{q^{\ell} w_i-z_j} \normord{\bm f(w_i) \Phi_0 (z_j)}
\]

 \item Finally, we get an extra term for the case $\e_i = j$:
\[
 \Phi_{b_j} (z_j) = \alpha_{b_j} \oint \ldots \oint \left(\prod_{i:\e_i=j} \frac{w_i dw_i}{2\pi i}\frac{1}{w_i-q^{\ell} z_j}\right) \Phi_0 (z_j) \prod_{i:\e_i=j} \bm f (w_i)
\]
where the last product is in the increasing order.
\end{itemize}

We conclude using lemma~\ref{prop:no_several}. We first write the case $k=0$:
\begin{multline}\label{eq:int}
 \Psi^{(0)}_{b_1,\ldots,b_{2n}} (z_1,\ldots,z_{2n}) = \gamma_{\ell,n} \prod_{i=1}^{2n} \alpha_{b_i}
 \prod_{1\le i<j\le 2n} (-z_i)^{\frac{\ell}{2}} \frac{\left(q^2 \frac{z_j}{z_i};q^4\right)_{\infty}}{\left(q^{2\ell+2} \frac{z_j}{z_i};q^4\right)_{\infty}}\\
 \times  \oint \ldots \oint \prod_{i=1}^{\ell n}\frac{w_i dw_i}{2\pi i}
 \prod_{i<j}\frac{w_j-w_i}{w_j-q^2 w_i} \frac{P_{Y_{\ell,n}} (w_1,\ldots, w_{\ell n})}{\prod_{j \leq \e_i} (q^{\ell} w_i-z_j) \prod_{j \geq \e_i}(w_i - q^{\ell}z_j)} 
\end{multline}
Where the countours are chosen such that
\begin{align*}
 |w_i|&<|q^{-\ell}z_j|\qquad\text{for all }j\leq \e_i &
 |w_i|&>|q^{\ell}z_j| \qquad\text{for all }j\geq \e_i\\
 |z_i | & > |q^2 z_j| \qquad\text{for all }i < j &
 |w_i|&>|q^{-2} w_j| \qquad\text{for all }i<j
\end{align*}
This choice is consistent with remark~\ref{rmk:normal_order}.

The same computation can be repeated for $\Psi_{b_1,\ldots,b_{2n}}^{(k)} = \braC{k} \Phi_{b_1} \ldots \Phi_{b_{2n}}\ketC{k}$, and we obtain:
\begin{multline}\label{eq:intb}
 \Psi_{b_1,\ldots,b_{2n}}^{(k)} (z_1,\ldots,z_{2n}) = \gamma_{\ell,n}^{(k)} \prod_{i=1}^{2n} \alpha_{b_i} (-z_i)^{\frac{k}{2}}
 \prod_{1\le i<j\le 2n} (-z_i)^{\frac{\ell}{2}} \frac{\left(q^2 \frac{z_j}{z_i};q^4\right)_{\infty}}{\left(q^{2\ell+2} \frac{z_j}{z_i};q^4\right)_{\infty}} \\
 \times \oint \ldots \oint \prod_{i=1}^{\ell n}\frac{dw_i}{2\pi i}
  \prod_{i<j}\frac{w_j-w_i}{w_j-q^2 w_i} \frac{P_{Y_{\ell,n}^{(k)}} (w_1,\ldots, w_{\ell n})}{\prod_{j \leq \e_i} (q^{\ell} w_i-z_j) \prod_{j \geq \e_i}(w_i - q^{\ell}z_j)} 
\end{multline}

In general, we do not expect a simple formula for the Macdonald polynomial $P_{Y_{\ell,n}^{(k)}}$, except in two cases:
\begin{enumerate}
\item
If $\ell=1$, the parafermionic part of the current is trivial, and 
$P_{Y_{1,n}^{(k)}} = \prod_i w_i^{1-k} \allowbreak \prod_{i < j} (w_i-q^2 w_j) (w_i - q^{-2} w_j)$, and we recover
the standard level $1$ formulae found in e.g.~\cite{JM-book}.
\item 
If $\ell=2$, the parafermionic field reduces to a ($q$-deformed) fermionic
field \cite{Ber-VOB}, and
$P_{Y_{2,n}^{(k)}} = \prod_{i < j} \frac{(w_i-q^2 w_j)(w_i-q^{-2}w_j)}{w_i-w_j} \Pf \left(\frac{(w_i-w_j)E_{2-k}(w_i,w_j)}{(w_i-q^2 w_j)(w_i-q^{-2} w_j)}\right)_{1\leq i,j\leq 2n}$, where $E_m(w_i,w_j)$ is the elementary symmetric polynomial, reproducing results of 
\cite{Idz}.
\end{enumerate}

In the rest of this section, we shall write $\Psi:=\Psi^{(k)}$ when
there is no risk of confusion.

\subsection{Quantum Knizhnik--Zamolodchikov equation}
A detailed proof of the $q$KZ equation can be found in \cite[appendix A]{IIJMNT}). For the convenience of the reader, we provide a short graphical proof
in appendix~\ref{app:qKZ}.

Define
\[
R_\pm(z)=\phi_\pm(z)\bar R(z)
\]
with
\begin{align*}
\phi_-(z)&=q^{-\ell^2/2}\frac{\poc{q^2z}^2}{\poc{q^{2\ell+2}z}\poc{q^{-2\ell+2}z}}
\\
\phi_+(z)&=\phi_-(1/z)^{-1}
\end{align*}

Note that compared to the other normalization $R(z)$ we have used before,
one has:
\[
R_\pm(z)=\tilde\phi_\pm(z) R(z)
\qquad
\tilde\phi_-(z)=
\begin{cases}
(-1)^m
&\ell=2m\text{ even}
\\
(-1)^m (z/q)^{1/2}
\frac{\poc{q^2z}\poc{q^2/z}}{\poc{z}\poc{q^4/z}}
&\ell=2m+1\text{ odd}
\end{cases}
\]
and $\tilde\phi_+(z)=\tilde\phi_-(1/z)^{-1}$.

The following equation is then
satisfied by 
$\Psi(z_1,\ldots,z_L)=\braC{k}\Phi(z_1)\ldots\Phi(z_L)\ketC{k}$:
\begin{multline}\label{eq:qKZ}
R_-(z_{i+1}/(sz_i))\ldots R_-(z_L/(sz_i))K_1^{1+k}
R_+(z_1/z_i)\ldots R_+(z_{i-1}/z_{i})
\Psi(z_1,\ldots,z_L)
\\=\Psi(z_1,\ldots,s\,z_i,\ldots,z_L)
\qquad i=1,\ldots,L
\end{multline}
where $s=q^{-2(\ell+2)}$; all the indices are suppressed and can be recovered
by following the spectral parameters $z_i$, see appendix~\ref{app:qKZ} for
details. This is the quantum Knizhnik--Zamolodchikov equation, first
introduced in \cite{FR-qKZ}.

As a consequence of \eqref{eq:exchPhi} for the VOs we also have
the following identity (exchange relation):
\begin{equation}\label{eq:exch}
\check R_{i,i+1}(z_{i}/z_{i+1})\Psi(z_1,\ldots,z_L)
=\Psi(z_1,\ldots,z_{i+1},z_i,\ldots,z_L)
\end{equation}
We shall come back to this equation in section~\ref{sec:poly}.

\subsection{Recurrence relations}
Define for convenience the ``dual vertex operator'' $\Phi^\ast$,
which by self-duality of evaluation representation can be expressed 
in terms of $\Phi$ as:
\begin{equation}\label{eq:dualvo}
\Phi_b^\ast(z)=c_b z^{\ell/2}\Phi_{\ell-b}(q^{-2}z)
\end{equation}
where $c_b=(-1)^{\nicefrac{\ell}{2}+b+k} q^{k(\ell-k)-b(\ell-b)} \left[ \ell \atop k \right] /\left[\ell\atop b\right]$.

From general arguments 
one can prove that
\begin{subequations}\label{eq:dualrel}
\begin{align}
\Phi_b(z)\Phi^\ast_{b'}(z)&=\delta_{b,b'} g_k
\\
\sum_{b=0}^\ell \Phi^\ast_b(z)\Phi_b(z)&= g_k
\end{align}
\end{subequations}
where $g_k=q^{k(\ell-k)} \left[\ell\atop k\right]
\frac{\poc{q^{2(\ell+1)}}}{\poc{q^2}}$.
cf~\cite[Prop.~4.1]{IIJMNT}.

This implies immediately the following recurrence relations 
for the correlation functions:
\begin{subequations}\label{eq:prerecrel}
\begin{align}\label{eq:prerecrela}
\Psi_{\ldots,b,b',\ldots}(\ldots,z,q^{-2}z,\ldots)&=
\frac{g_k}{c_b} \delta_{b+b',\ell}
z^{-\ell/2}\,
\Psi_{\ldots}(\ldots)
\\\label{eq:prerecrelb}
\sum_{b=0}^\ell 
c_b
\Psi_{\ldots,\ell-b,b,\ldots}(\ldots,q^{-2}z,z,\ldots)&=g_k
z^{-\ell/2}\,
\Psi_{\ldots}(\ldots)
\end{align}
\end{subequations}
where $\dots$ are omitted (arbitrary) arguments.


\subsection{Polynomiality, cyclicity}\label{sec:poly}
%
It is convenient to redefine
\begin{equation}\label{eq:defpol}
\poly\Psi^{(k)}(z_1,\ldots,z_L)
=
\prod_{i=1}^L (-z_i)^{(L-1)\ell/4 + k/2}
\prod_{1\le i<j\le L} F(z_j/z_i)\, 
\Psi^{(k)}(z_1,\ldots,z_L)
\end{equation}
where
\[
F(z)=z^{-\ell/4} \frac{\poc{q^4z}}{\poc{q^{2(\ell+2)}z}}
\]
Note that the prefactor
$\prod_{i=1}^L z_i^{(L-1)\ell/4 + k/2}$ combined with equation~\eqref{eq:dimPhi}
implies that \linebreak $\poly\Psi(xz_1,\ldots,xz_L)=x^{\ell n(n-1) + kn}\poly\Psi(z_1,\ldots,z_L)$
(where $L=2n$);
in fact, rewriting the integral formula~\eqref{eq:int} as:
\begin{equation}\label{eq:Psi_pol}
 \begin{aligned}
 \poly\Psi^{(k)}_{b_1,\ldots,b_{2n}} (z_1,\ldots,z_{2n}) &=\gamma_{\ell,n}^{(k)} \prod_{i=1}^{2n} \alpha_{b_i} z_i^k 
\prod_{1\le i<j\le 2n} \prod_{r=1}^\ell (q^{2r}z_j-z_i)
\oint \ldots \oint \prod_{i=1}^{\ell n}\frac{dw_i}{2\pi i}\\
 & \quad \times \prod_{i<j}\frac{w_j-w_i}{w_j-q^2 w_i} \frac{P_{Y_{\ell,n}^{(k)}} (w_1,\ldots, w_{\ell n})}{\prod_{j \leq \e_i} (q^{\ell} w_i-z_j) \prod_{j \geq \e_i}(w_i - q^{\ell}z_j)} 
 \end{aligned}
\end{equation}
(where we recall that $\bm\e=(\e_1,\ldots,\e_{\ell n})$ is the sequence such that
$b_j = \#\{i \ \text{such that }\e_i=j\}$)
one can show (see appendix \ref{app:poly})
that $\poly\Psi^{(k)}$ is a (vector-valued) {\em polynomial}
in the variables $z_1,\ldots,z_L$ of degree $\ell n(n-1)+kn$.
Furthermore, its coefficients are rational functions of $q$
whose denominators are products of $1-q^{2j}$, $j=1,\ldots,\ell$.

For example, when $\ell=2$ and $n=2$ there are nineteen components, among which:
\begin{align*}
 \bm \Psi^{(0)}_{1,2,1,0} (z_1,z_2,z_3,z_4)&= (q+q^3) (z_1-q^2 z_2)(z_2-q^2 z_3)(z_3-q^2 z_4)(z_1-q^6 z_4)\\
 \bm \Psi^{(0)}_{1,2,0,1} (z_1,z_2,z_3,z_4)&= -(q+q^3) (z_1-q^2 z_2)(z_3-q^2 z_4) \\
 & \quad \times(q^2 z_1 z_2+z_1 z_3-q^4 z_1 z_3 - q^4 z_2 z_3 - q^4 z_1 z_4 + q^{10} z_3 z_4)
\end{align*}
For generic $n$ and $\ell$ the simplest component is the one with $b_i = 0$ for all $i\leq n$ and $b_i = \ell$ for the remainings $i > \ell$:
\begin{equation}\label{eq:simplecomp}
 \bm \Psi^{(k)}_{0, \ldots,0,\ell,\ldots,\ell} = (-1)^{\ell \frac{n(n+1)}{2}} q^{(\ell-k) n + \ell n \frac{n-\ell}{2}} \prod_{i=1}^n z_i^k \prod_{1 \leq i < j \leq n} \prod_{r=1}^{\ell} (q^{2r}z_j-z_i)(q^{2r}z_{n+j}-z_{n+i})
\end{equation}
and all its rotations.


Now define yet another normalization of the $R$-matrix, namely
$\poly R(z)=\frac{F(z)}{F(1/z)} R(z)$,
which is adapted to $\poly\Psi$ from \eqref{eq:defpol}.
Explicitly,
\begin{equation}\label{eq:defRp}
\check {\poly R}(z)=\sum_{j=0}^\ell \prod_{r=1}^j 
\frac{1-q^{2r}z}{z-q^{2r}}P_{j}
\end{equation}

The exchange relation \eqref{eq:exch} becomes
\begin{equation}\label{eq:exchp}
\check{\poly R}_{i,i+1}(z_{i}/z_{i+1})\poly\Psi(z_1,\ldots,z_L)
=\poly\Psi(z_1,\ldots,z_{i+1},z_i,\ldots,z_L)
\end{equation}
Starting from \eqref{eq:qKZ} at say $i=1$, switching to $\poly\Psi$,
applying repeatedly \eqref{eq:exchp} and using $\tilde\phi_+(1/z)F(sz)/F(1/z)=(-1)^\ell$,
we find the cyclicity relation:
\begin{equation}\label{eq:cycl}
(-1)^\ell s^{(n-1)\ell/2} q^{(1+k)a_{0,L}} S\poly\Psi(z_L,z_1,\ldots,z_{L-1})=\poly\Psi(z_1,z_2,\ldots,sz_L)
\end{equation}
where $S: 
V_{z_L}\otimes V_{z_1}\otimes\cdots\otimes V_{z_{L-1}}
\to 
V_{z_1}\otimes\cdots\otimes V_{sz_L}
$ 
rotates cyclically the tensor product.

The system of equations (\ref{eq:exchp}--\ref{eq:cycl}) is similar
to the one that appeared first in \cite{Smi} in the study of form factors.
Note that the power of $s$ is a reflection of the homogeneity of
$\poly\Psi(z_1,\ldots,z_L)$. In components, \eqref{eq:cycl} writes:
\[
(-1)^\ell s^{(n-1)\ell/2} q^{(1+k)(2b_L-\ell)} \poly\Psi_{b_L,b_1,\ldots,b_{L-1}}(z_L,z_1,\ldots,z_{L-1})=\poly\Psi_{b_1,\ldots,b_L}(z_1,z_2,\ldots,sz_L)
\]

Other formulae can be rewritten in terms of $\poly\Psi$ as well;
in particular, the recurrence relations \eqref{eq:prerecrel} become:
\begin{subequations}\label{eq:reqrel}
\begin{multline}\label{eq:recrela}
\poly\Psi_{\ldots,b,b',\ldots}(\ldots,z_i,z_{i+1}=q^{-2} z_i,\ldots)=
(-1)^{k-\nicefrac{\ell}{2}} \frac{1}{c_b} \delta_{b+b',\ell} \left[ \ell \atop k \right] q^{k(\ell-k-1)}
\\
z_i^k\,
\prod_{j=1}^{i-1} \prod_{r=1}^\ell(z_j-q^{2r}z_i)
\prod_{j=i+2}^L \prod_{r=1}^\ell(z_j-q^{-2(r+1)}z_i)
\poly\Psi_{\ldots}(\ldots)
\end{multline}
\begin{multline}\label{eq:recrelb}
\sum_{b=0}^\ell c_b
\poly\Psi_{\ldots,\ell-b,b,\ldots}(\ldots,\,
z_{i}=q^{-2} z_{i+1},z_{i+1},\ldots)=
(-1)^{k-\nicefrac{\ell}{2}} \left[ \ell \atop k \right] q^{k(\ell - k +1)}
\\
z_i^k\,
\prod_{j=1}^{i-1} \prod_{r=1}^\ell(z_j-q^{2r}z_{i+1})
\prod_{j=i+2}^L \prod_{r=1}^\ell(z_j-q^{-2(r+1)}z_{i+1})
\poly\Psi_{\ldots}(\ldots)
\end{multline}
\end{subequations}

\subsection{Wheel condition}
Finally, we point out that $\poly\Psi$ satisfies a wheel condition
of a different nature than 
that of the Macdonald polynomials of sect.~\ref{sec:corrpara}. 
Namely, we have the following theorem:
\begin{theorem}\label{thm:wheel}
Assume that three parameters $z_{i_1}$, $z_{i_2}$, $z_{i_3}$,
$i_1<i_2<i_3$, form a ``wheel'': $z_i=q^{2(a+b)}z$, $z_j=q^{2b}z$, $z_k=z$
with $a,b$ integers such tht $a,b>0$ and $a+b<\ell+2$. Then
\begin{equation}\label{eq:wheel}
\poly\Psi(\ldots,q^{2(a+b)}z,\ldots,q^{2b}z,\ldots,z,\ldots)=0
\end{equation}
\end{theorem}
The proof is similar to the one given in \cite[theorem 3]{artic37}
and is reproduced in appendix~\ref{sec:proofwheel}.

Note furthermore that the entries of $\poly\Psi$ are {\em nonsymmetric}\/
polynomials, hence the ordering condition for its arguments in the wheel
condition. This wheel condition is not a special case
of the one considered in \cite{Kasa-wheel} for nonsymmetric Macdonald
polynomials.

\section{Application to integrable spin chains}\label{sec:applic}
In all the discussion that precedes concerning vertex operators,
it was assumed that $|q|<1$.
However, we have seen in section~\ref{sec:poly}
that one can get rid of all infinite products in explicit formulae
by a redefinition of the solution $\poly\Psi$ of $q$KZ and of
the $R$-matrix $\poly R$. The former 
becomes a polynomial of $z_1,\ldots,z_L$ (with coefficients
of the form $P(q)/\prod_{i=1}^\ell (1-q^{2i})^{k_i}$) 
whereas the latter
becomes a rational function of them. As a result, we can now relax
the constraint $|q|<1$ and in particular consider the case where
$q^2$ is a primitive $(\ell+2)^{\text{th}}$ root of unity,
so that $s=q^{-2(\ell+2)}=1$. It is the purpose of this section to show
that $\poly\Psi$ then becomes an eigenvector of an integrable 
transfer matrix. 

Note that since $\poly\Psi\sim \braC{k}
\Phi\ldots\Phi\ketC{k}$, our procedure is similar to a ``matrix product
Ansatz'', where the role of the
matrix algebra is played here by the Zamolodchikov--Faddeev 
algebra \cite{ZZ-FZalg,Fad-FZalg,Luk-FZalg}.
In fact, a similar procedure has already been proposed in \cite{AL-MPA,AL-MPA2,KM-MPA}.
However, there is a crucial difference between our situation and theirs.
In the aforementioned papers, in order to impose periodic boundary conditions,
a trace is taken, which would correspond to computing
$\tr(\Phi\ldots\Phi)$. Unfortunately such a trace is divergent in our
setting. Also observe that such an Ansatz would produce
eigenvectors of integrable models for an arbitrary
value of the quantum parameter $q$, 
a claim which we do not make here. Instead we take
a vacuum expectation value of VOs; the price to pay is that rotational
invariance (periodic boundary conditions) is only restored at special roots
of unity, namely, $q^{2(\ell+2)}=1$. 
One can also regularize the trace by
adding a $x^d$: $\tr(x^d \Phi\ldots\Phi)$, producing finite temperature
correlation functions, as already mentioned in the introduction;
but that would also spoil the rotational invariance, and is therefore
not directly relevant in the present context.

\subsection{The model}\label{sec:spinchain}
Define the inhomogeneous monodromy matrix to be an operator
on $V_z\otimes V_{z_1}\otimes\cdots\otimes V_{z_L}$ of the form:
\[
\mathcal{T}(z)=\poly R_{10}(z_1/z)\ldots \poly R_{L0}(z_L/z)
\]
where the $0$ index corresponds to $V_z$ and non zero indices to $V_{z_i}$,
and the $R$-matrix is given as before by \eqref{eq:defRp}.
with $\check{\poly R}=\mathcal{P}\poly R$.

The twisted transfer matrix is then defined by taking
the trace over the auxiliary space $V_z$, with a twist
(in which it is convenient to absorb the sign $(-1)^\ell$):
\begin{equation}\label{eq:defT}
\poly T(z)=\tr_0 ((-q^{1+k})^{a_0} \mathcal{T}(z))
\end{equation}

The Yang--Baxter equation implies that transfer matrices commute for different spectral parameters:
\[
[\poly T(z),\poly T(z')]=0
\]

The diagonalization of $\poly T(z)$ is usually performed using Bethe Ansatz,
which produces eigenvectors defined in terms of ``Bethe roots'', i.e., 
solutions of certain algebraic equations. We do not pursue this route here.

A local Hamiltonian can be extracted from the transfer matrix
in the homogeneous limit where all the $z_i$ coincide.
Assume 
that $z_i=1$.
According to \eqref{eq:defRp}, $R(1)=\mathcal{P}$, and therefore
\begin{equation}\label{eq:T1}
\poly T(1)=(-q^{1+k})^{a_{0,L}} S
\end{equation}
where we recall that $S$ is cyclic permutation of factors of the tensor product.
$\poly T(1)$ is a discrete analogue of the momentum operator.
Expanding to next order, we obtain the Hamiltonian:
\[
\poly H=\frac{1}{i}\poly T(1)^{-1}\frac{d\poly T}{dz}|_{z=1}
=\sum_{j=1}^L h_{j,j+1}
\]
where 
$h_{j,j+1}=\frac{1}{i}\frac{d}{dz}\check {\poly R}_{j,j+1}(z)|_{z=1}$, $i<L$,
and $h_{L,1}=\frac{1}{i}(-q^{1+k})^{-a_{0,L}}
\frac{d}{dz}\check {\poly R}_{L,1}(z)|_{z=1}
(-q^{1+k})^{a_{0,L}}$ 
(twisted periodic boundary conditions). The factor of $1/i$ is introduced
for convenience, see below.

\subsection{The simple eigenvalue}\label{sec:eig}
Assume that $q^2$ is primitive $(\ell+2)^{\text{th}}$ root of unity
(primitiveness being necessary for $\poly\Psi$ to be well-defined).
Consider $\poly\Psi$ given by
\eqref{eq:Psi_pol}. An important remark is that 
the values of $q_M=q^{2(\ell+1)}$ and $t_M=q^{-2}$ coincide,
so that the Macdonald polynomial appearing in this expression 
becomes simply the corresponding 
{\em Schur polynomial}. Therefore,
\begin{equation}\label{eq:Psi_polb}
 \begin{aligned}
 \poly\Psi^{(k)}_{b_1,\ldots,b_{2n}} (z_1,\ldots,z_{2n}) &=\gamma_{\ell,n}^{(k)} \prod_{i=1}^{2n} \alpha_{b_i} z_i^k 
\prod_{1\le i<j\le 2n} \prod_{r=1}^\ell (q^{2r}z_j-z_i)
\oint \ldots \oint \prod_{i=1}^{\ell n}\frac{dw_i}{2\pi i}\\
 & \quad \times 
\frac{\det(w_i^{h_{j-1}})_{i,j=1,\ldots,2n}}
{\prod_{i<j}(w_j-q^2 w_i)
\prod_{j \leq \e_i} (q^{\ell} w_i-z_j) \prod_{j \geq \e_i}(w_i - q^{\ell}z_j)}
 \end{aligned}
\end{equation}
where $h_j=\lfloor \frac{j}{\ell}\rfloor+\lfloor \frac{j+\ell-k}{\ell}\rfloor+j$.

Next, note that \eqref{eq:cycl} can be simplified if one assumes $s=q^{-2(\ell+2)}=1$ to
\begin{equation}\label{eq:cyclb}
(-q)^{(1+k)a_{0,L}}
S\poly\Psi(z_2,\ldots,z_L,z_1)
=\poly\Psi(z_1,\ldots,z_L)
\end{equation}

Now consider the effect of the transfer matrix on the correlation function
$\poly\Psi(z_1,\ldots,z_L)$. Writing the transfer matrix $\mathbf T(z)$ as the trace of the monodromy matrix $\mathcal{T}(z)$ and using
\eqref{eq:recrela}, we have\footnote{The argument that follows was suggested
 to us by R.~Weston, to whom we are indebted.}
\begin{equation}\label{eq:RWtrick}
\poly T(z)\poly\Psi(z_1,\ldots,z_L)
=
\Pi^{-1}\sum_{b,b'} c_b ((-q^{1+k})^{a_0})_{b,b} \mathcal{T}_{b,b'}(z) 
\poly\Psi_{b,\ell-b',\ldots}(z,q^{-2}z,z_1,\ldots,z_L)
\end{equation}
where $\Pi=\left[{\ell\atop k}\right]z_i^k\prod_{j=1}^{i-1} \prod_{r=1}^\ell(z_j-q^{2r}z_i)
\prod_{j=i+2}^L \prod_{r=1}^\ell(z_j-q^{-2(r+1)}z_i)$
is the factor appearing in \eqref{eq:recrela},
and the $\ldots$ in subscript
mean that only the first and the last index are fixed,
the rest forming a vector in $V_{z_1}\otimes\cdots\otimes V_{z_L}$.

Using \eqref{eq:cyclb}, we can rewrite this as
\[
\poly T(z)\poly\Psi(z_1,\ldots,z_L)
=
\Pi^{-1}\sum_{b,b'} c_b \mathcal{T}_{b,b'}(z)
\poly\Psi_{\ell-b',\ldots,b}(q^{-2}z,z_1,\ldots,z_L,z)
\]

Finally, 
writing $\mathcal{T}(z)$ as a product of $\poly R$-matrices and applying the
exchange relation \eqref{eq:exchp} repeatedly, we find:
\begin{align*}
\poly T(z)\poly\Psi(z_1,\ldots,z_L)
&=\Pi^{-1}\sum_b c_b
\poly\Psi_{\ell-b,b,\ldots}(q^{-2}z,z,z_1,\ldots,z_L)
\\
&=\poly\Psi(z_1,\ldots,z_L)
\end{align*}
where the last equality follows from \eqref{eq:recrelb}.

Noting that $\poly\Psi\ne0$ because of \eqref{eq:simplecomp},
we conclude that if $q^2$ is a primitive
$(\ell+2)^{\text{th}}$ root of unity,
$\poly\Psi(z_1,\ldots,z_L)$ is an eigenvector of
the (inhomogeneous, twisted) transfer matrix $\poly T(z)$,
with a trivial eigenvalue.
A graphical interpretation of this proof can be found in appendix 
\ref{sec:grapheig}.

The twist is $-q^{1+k}$,
but note that only its square is meaningful.
As $k$ varies from $0$ to $\ell$, $q^{2(1+k)}$ spans all
$\ell+1$ nontrivial $(\ell+2)^{\text{th}}$ roots of unity. 

We conjecture that if $q=-e^{\pm i\pi/(\ell+2)}$ (for all
twists $-q^{1+k}$), the eigenvector we have
just constructed corresponds to the largest eigenvalue of $\mathbf{T}(z)$
for $z_i$ of modulus $1$ and sufficiently close to $1$. In particular,
for $z_i=1$, we conjecture that $\poly\Psi$ is the {\em ground state}\/
eigenvector of the Hamiltonian $\poly H$ (with zero ground state energy). 
For a discussion of the latter statement, see sect.~\ref{sec:susy}.
Note that at $q=\pm e^{\pm i\pi/(\ell+2)}$, the spectrum of $\poly H$ 
is real.

\rem{odd vs even twisted?}

\subsection{Relation to loop model}\label{sec:loop}
In \cite{artic37}, a ``higher spin'' loop model was built by fusing
the standard Temperley--Lieb loop model. Since the latter is related
to the spin $1/2$ representation of 
$U=U_q(\widehat{\mathfrak{sl}(2)})$ in the sense that it is equivalent
to the XXZ spin chain/6-vertex model, the fused loop model must be related
to higher spin integrable Hamiltonians/transfer matrices.
The only nontrivial issue is that of boundary conditions. Let us discuss this
here, first in the spin $1/2$ case, then in the fused case.

\subsubsection{Equivalence to spin model and twisted boundary conditions}
\label{sec:loopequiv}
%
%
%
\makeatletter
\newcommand{\gettikzxy}[3]{
  \tikz@scan@one@point\pgfutil@firstofone#1\relax
\pgfmathsetmacro{#2}{\the\pgf@x/\linkpatternunit}
\pgfmathsetmacro{#3}{\the\pgf@y/\linkpatternunit}
}
\tikzset{label anchor/.code={%
    \let\tikz@auto@anchor=\pgfutil@empty
    \def\tikz@anchor{#1}
  },
  label anchor/.default=center
}
\makeatother
\newdimen\linkpatternunit%
\newcount\linkpatternsize%
\newcount\lpsize
%
\newif\iflinkpatterninverted
\newif\iflinkpatterntikzstarted%
\newif\iflinkpatternboxed
\newif\iflinkpatternaxis
\newif\iflinkpatternstraightlines
\newif\iflinkpatternnumbered
%
\newcount\linkpatternfused
%
%
\pgfkeys{/linkpattern/.cd,inverted/.is if=linkpatterninverted,numbered/.is if=linkpatternnumbered,tikzstarted/.is if=linkpatterntikzstarted,straight lines/.is if=linkpatternstraightlines,boxed/.is if=linkpatternboxed,axis/.is if=linkpatternaxis,vertexcolor/.store in=\linkpatternvertexcolor,edgecolor/.store in=\linkpatternedgecolor,boxcolor/.code={\linkpatternboxedtrue\def\linkpatternboxcolor{#1}},tikzoptions/.style={every linkpattern/.append style={#1}},size/.code={\linkpatternsize=#1},numbering/.code={\linkpatternnumberedtrue\def\linkpatternnumbering{#1}\def\lpnumbering{#1}},unit/.code={\linkpatternunit=#1},height/.store in=\linkpatternheight,shape/.store in=\linkpatternshape,looseness/.store in=\linkpatternlooseness,squareness/.store in=\linkpatternsquareness,width/.store in=\linkpatternwidth,
pipedream/.style={shape=pipedream,looseness=0,straight lines,numbering=tangle,numbered=false},
tangle/.style={shape=tangle,numbering=tangle,numbered=false},
every linkpattern/.style={x=\linkpatternunit,y=\linkpatternunit,baseline=0},
fused/.code={\linkpatternfused=#1}}
\linkpatterninvertedfalse%
\linkpatternnumberedfalse%
\linkpatterntikzstartedfalse%
\linkpatternboxedfalse%
\linkpatternaxistrue%
\linkpatternunit=0.6cm%
\linkpatternsize=0%
\linkpatternfused=1
\linkpatternstraightlinesfalse%
\def\linkpatternlooseness{0.2}
\def\linkpatternsquareness{0.4}
\def\linkpatternvertexcolor{red}%
\def\linkpatternedgecolor{blue}%
\def\linkpatternboxcolor{none}%
\def\linkpatternheight{0}
\def\linkpatternwidth{0}
\def\linkpatternshape{default}
\def\linkpatternnumbering{default}
\tikzset{vertex/.style={circle,thin,draw=black,fill=\linkpatternvertexcolor,inner sep=1.5pt}}%
\tikzset{edge/.style={very thick,draw=\linkpatternedgecolor}}%
\def\firstchar#1#2\empty{#1}%
\def\linkpatterndo#1#2{
\edef\param{\csname linkpattern#2\endcsname}
\edef\firstcharparam{\expandafter\firstchar\param\empty}
\expandafter\ifcat\firstcharparam a
\expandafter\ifx\csname linkpattern#1\param\endcsname\relax
\csname linkpattern#1unknown\endcsname
\else
\csname linkpattern#1\csname linkpattern#2\endcsname\endcsname
\fi
\else
\csname linkpattern#1unknown\endcsname
\fi
}%
\def\linkpatterncoorddefault{\xdef\lpcoordx{\x}\xdef\lpcoordy{0}\xdef\lpangle{90}}%
\def\linkpatterncoordtangle{\ifnum\x>\lphalfsize\pgfmathparse{\lpsize+1-\x}\xdef\lpcoordx{\pgfmathresult}\xdef\lpcoordy{\lpheight}\xdef\lpangle{270}\else\xdef\lpcoordx{\x}\xdef\lpcoordy{-\lpheight}\xdef\lpangle{90}\fi}
\def\linkpatterncoordcircle{%
 \pgfmathparse{180+(\x-1)*360/\lpsize}
 \xdef\lpangle{\pgfmathresult}
 \pgfmathparse{-cos(\lpangle)}
\xdef\lpcoordx{\pgfmathresult}
 \pgfmathparse{-sin(\lpangle)}
 \xdef\lpcoordy{\pgfmathresult}
}%
\def\linkpatterncoordpipedream{\ifnum\x>\lphalfsize\pgfmathparse{\lpsize+1-\x-0.5}\xdef\lpcoordx{\pgfmathresult}\xdef\lpcoordy{0}\xdef\lpangle{270}\else\pgfmathparse{0.5-\x}\xdef\lpcoordy{\pgfmathresult}\xdef\lpcoordx{0}\xdef\lpangle{0}\fi}
\def\linkpatterncoordrectangle{
\ifnum\x>\lptqsize
\pgfmathparse{\lpsize+1-\x-0.5}\xdef\lpcoordx{\pgfmathresult}\xdef\lpcoordy{0}\xdef\lpangle{270}
\else\ifnum\x>\lphalfsize
\pgfmathparse{\x-\lptqsize-0.5}\xdef\lpcoordy{\pgfmathresult}\xdef\lpcoordx{\linkpatternwidth}\xdef\lpangle{180}
\else\ifnum\x>\linkpatternheight
\pgfmathparse{\x-\linkpatternheight-0.5}\xdef\lpcoordx{\pgfmathresult}\xdef\lpcoordy{-\linkpatternheight}\xdef\lpangle{90}
\else
\pgfmathparse{0.5-\x}\xdef\lpcoordy{\pgfmathresult}\xdef\lpcoordx{0}\xdef\lpangle{0}
\fi\fi\fi
}%
\def\linkpatterncoordunknown{\message{link pattern: unknown shape}}%
%
\def\linkpatterndrawaxisdefault{\draw (1,0) -- (\lpsize,0);}%
\def\linkpatterndrawaxistangle{\draw (1,-\lpheight) -- (\lphalfsize,-\lpheight) (1,\lpheight) -- (\lphalfsize,\lpheight);}%
\def\linkpatterndrawaxiscircle{\draw (0,0) circle (1);}%
\def\linkpatterndrawaxispipedream{\draw (0,-\lphalfsize) -- (0,0) -- (\lphalfsize,0);}%
\def\linkpatterndrawaxisrectangle{\draw (0,-\linkpatternheight) rectangle (\linkpatternwidth,0);}%
\def\linkpatterndrawaxisunknown{\message{link pattern: unknown shape}}%
%
\def\linkpatterndrawboxdefault{\draw[fill=\linkpatternboxcolor] (0.5,0) rectangle (\lpsize+0.5,\lpheight);}%
\def\linkpatterndrawboxtangle{\draw[fill=\linkpatternboxcolor] (0.5,-\lpheight) rectangle (\lphalfsize+0.5,\lpheight);}
\def\linkpatterndrawboxcircle{\draw[fill=\linkpatternboxcolor] (0,0) circle (1);}%
\def\linkpatterndrawboxunknown{\message{link pattern: unknown shape}}%
\def\linkpatterndrawboxpipedream{\draw[fill=\linkpatternboxcolor] (0,-\lphalfsize) rectangle (\lphalfsize,0);}%
\def\linkpatterndrawboxrectangle{\draw[fill=\linkpatternboxcolor] (0,-\linkpatternheight) rectangle (\linkpatternwidth,0);}%
%
\def\linkpatternsetsizeunknown{
\global\lpsize=\linkpatternsize
\if\linkpatternheight0
\xdef\maxsep{0}
\foreach \x/\xx in \mylist%
{%
\edef\tempx{\withoutprime{\x}}
\edef\tempxx{\withoutprime{\xx}}
\pgfmathparse{max(\maxsep,abs(\tempx-\tempxx))}
\xdef\maxsep{\pgfmathresult}
}%
\pgfmathparse{0.25+0.8*\linkpatternsquareness*\maxsep}
\xdef\lpheight{\pgfmathresult}
\else
\xdef\lpheight{\linkpatternheight}
\fi
}
\def\linkpatternsetsizerectangle{
\pgfmathtruncatemacro{\tempsize}{2*\linkpatternwidth+2*\linkpatternheight}
\global\lpsize=\tempsize
\pgfmathtruncatemacro{\tempsize}{\linkpatternwidth+2*\linkpatternheight}
\xdef\lptqsize{\tempsize}
}%
%
\newcount\tempsize
\def\linkpatternrightmostunknown{
\global\lpsize=0
\global\tempsize=0
\foreach\x/\labx in \linkpatternnumbering
{
\edef\tempx{\withoutprime{\x}}
\ifnum\lpsize<\tempx\global\lpsize=\tempx\fi
\global\advance\tempsize by 1
}
\ifnum\tempsize>\lpsize\global\lpsize=\tempsize\fi
}%
\def\linkpatternrightmostdefault{
\global\lpsize=0
\global\tempsize=0
\foreach \x/\y in \mylist
{
\edef\tempx{\withoutprime{\x}}
\ifnum\lpsize<\tempx\global\lpsize=\tempx\fi
\ifx\x\y
\global\advance\tempsize by 1
\else
\edef\tempy{\withoutprime{\y}}
\ifnum\lpsize<\tempy\global\lpsize=\tempy\fi%
\global\advance\tempsize by 2
\fi
}
\ifnum\tempsize>\lpsize\global\lpsize=\tempsize\fi
}%
\def\linkpatternrightmosttangle{
\global\lpsize=0
\global\tempsize=0
\foreach \x/\y in \mylist
{
\edef\tempx{\withoutprime{\x}}
\ifnum\lpsize<\tempx\global\lpsize=\tempx\fi
\ifx\x\y
\global\advance\tempsize by 1
\else
\edef\tempy{\withoutprime{\y}}
\ifnum\lpsize<\tempy\global\lpsize=\tempy\fi%
\global\advance\tempsize by 2
\fi
}
\global\advance\lpsize by\lpsize
\ifnum\tempsize>\lpsize\global\lpsize=\tempsize\fi
}%
\def\linkpatternrightmosthalftangle{\linkpatternrightmosttangle}
%
\def\linkpatternnumberingdefault{\xdef\lpnumbering{1,...,\lpsize}}%
\def\linkpatternnumberingtangle{%
\xdef\lpnumbering{1,...,\lphalfsize}
\foreach\x in {1,...,\lphalfsize}
{\xdef\lpnumbering{\lpnumbering,\x'}}
}%
\def\linkpatternnumberinghalftangle{%
\xdef\lpnumbering{1,...,\lphalfsize}
\foreach\x in {1,...,\lphalfsize}
{\xdef\lpnumbering{\lpnumbering,\x'/}}
}%
\def\linkpatternnumberingunknown{%
}%
\newcommand\linkpattern[2][]{
{
\pgfkeys{/linkpattern/.cd,#1}
\edef\mylist{#2}
\def\primetest##1'{}%
\def\hasaprime##1{\expandafter\primetest##1''}
\def\internalwithoutprime##1'{##1}%
\def\withoutprime##1{\if\hasaprime##1 %
\expandafter\internalwithoutprime##1\else ##1\fi}%
\iflinkpatternnumbered%
\iflinkpatterninverted
\tikzset{lbl/.style 2 args={label={[label anchor=-##1+180,inner sep=2pt]##1:$\scriptstyle##2$}}}
\else
\tikzset{lbl/.style 2 args={label={[label anchor=##1+180,inner sep=2pt]##1:$\scriptstyle##2$}}}
\fi
\else%
\tikzset{lbl/.style={}}
\fi%
\iflinkpatterntikzstarted
\begin{scope}[/linkpattern/every linkpattern]
\else%
\begin{tikzpicture}[/linkpattern/every linkpattern]%
\fi%
\iflinkpatterninverted%
\begin{scope}[yscale=-1]%
\fi%
\linkpatterndo{setsize}{shape}
\ifnum\lpsize=0
\linkpatterndo{rightmost}{numbering}
\fi
\pgfmathtruncatemacro{\lphalfsize}{\lpsize/2}
\linkpatterndo{numbering}{numbering}
\iflinkpatternboxed
\linkpatterndo{drawbox}{shape}
\else
\iflinkpatternaxis
\linkpatterndo{drawaxis}{shape}
\fi
\fi
\foreach\x/\xlab in \lpnumbering
{
\if\hasaprime\x %
\pgfmathtruncatemacro{\x}{\lpsize+1-\withoutprime{\x}}
\fi
\linkpatterndo{coord}{shape}
\path (\lpcoordx,\lpcoordy) coordinate[transform shape,vertex,lbl={\lpangle+180}{\xlab},alias=v\xlab] (v\x)
 ++(\lpangle:\linkpatternunit) coordinate[alias=vv\xlab] (vv\x);
}
\foreach \a/\b/\c in \mylist
{
\if\hasaprime\a %
\pgfmathtruncatemacro{\a}{\lpsize+1-\withoutprime{\a}}
\fi
\draw[edge]
\ifx\b\c
\else
[decoration={markings,mark = at position 0.5 with { \arrow[semithick]{\c} }},postaction={decorate}]
\fi
\ifx\a\b
(v\a) -- ++(0,\lpheight);
\else
\pgfextra{
\if\hasaprime\b %
\pgfmathtruncatemacro{\b}{\lpsize+1-\withoutprime{\b}}
\fi
\gettikzxy{(v\a)}{\ax}{\ay}
\gettikzxy{(v\b)}{\bx}{\by}
\gettikzxy{(vv\a)}{\axx}{\ayy}
\gettikzxy{(vv\b)}{\bxx}{\byy}
\pgfmathsetmacro{\dist}{sqrt((\ax-\bx)*(\ax-\bx)+(\ay-\by)*(\ay-\by))}
\pgfmathsetmacro{\abx}{(\axx-\ax)*\dist*\linkpatternsquareness+(\bx-\ax)*\linkpatternlooseness)}
\pgfmathsetmacro{\aby}{(\ayy-\ay)*\dist*\linkpatternsquareness+(\by-\ay)*\linkpatternlooseness)}
\pgfmathsetmacro{\bax}{(\bxx-\bx)*\dist*\linkpatternsquareness+(\ax-\bx)*\linkpatternlooseness)}
\pgfmathsetmacro{\bay}{(\byy-\by)*\dist*\linkpatternsquareness+(\ay-\by)*\linkpatternlooseness)}
}
(v\a) 
\iflinkpatternstraightlines
\pgfextra{
\pgfmathsetmacro{\t}{((\ax-\bx)*\bay-(\ay-\by)*\bax)/(\aby*\bax-\abx*\bay)}
\pgfmathsetmacro{\abx}{\t*\abx}
\pgfmathsetmacro{\aby}{\t*\aby}
}
[rounded corners] -- ++(\abx,\aby) -- (v\b);
\else
.. controls ++(\abx,\aby) and ++(\bax,\bay) .. 
\fi
(v\b);
\fi
}
\iflinkpatterninverted
\end{scope}
\fi
\iflinkpatterntikzstarted
\end{scope}
\else%
\end{tikzpicture}%
\fi%
}}%
%
%
\newcommand\tanglelinkpattern[3][]{%
{
\pgfkeys{/linkpattern/.cd,#1}
\begin{tikzpicture}[/linkpattern/every linkpattern]%
\linkpattern[#1,tikzstarted,numbered=false]{#3}
\iflinkpatterninverted
\begin{scope}[yshift=0.5*\linkpatternunit]
\else
\begin{scope}[yshift=-0.5*\linkpatternunit]
\fi
\pgfmathtruncatemacro{\lptempsize}{2*\linkpatternsize}
\linkpattern[tangle,#1,tikzstarted,size=\lptempsize,
numbering=halftangle,
height=0.5]{#2}
\end{scope}
\end{tikzpicture}%
}}
%
%
\newcommand\diag[4][]{%
\pgfkeys{/linkpattern/.cd,#1}
\iflinkpatterntikzstarted\else%
\begin{tikzpicture}[scale=0.5]
\fi%
\iflinkpatterninverted%
\begin{scope}[yscale=-1]%
\fi%
\draw (0,0) grid (#2,#3);
\edef\mylist{#4}
\foreach\y/\x/\z in \mylist
{
\ifx\x\z
\draw[decorate,decoration={zigzag,
amplitude=1pt,segment length=5pt}]
(\x-0.5,#3) -- (\x-0.5,\y-0.5) node[circle,fill=black,inner sep=2pt] {} -- (#2,\y-0.5);
\else
\node at (\x-0.5,\y-0.5) {$\z$};
\fi
}
\iflinkpatterninverted
\end{scope}
\fi
\iflinkpatterntikzstarted\else%
\end{tikzpicture}%
\fi%
}
%
\makeatletter
\tikzset{circle split part fill/.style  args={#1,#2}{%
 alias=tmp@name,
  postaction={%
    insert path={
     \pgfextra{%
     \pgfpointdiff{\pgfpointanchor{\pgf@node@name}{center}}%
                  {\pgfpointanchor{\pgf@node@name}{east}}%
     \pgfmathsetmacro\insiderad{\pgf@x}
      \fill[#1] (\pgf@node@name.base) ([xshift=-\pgflinewidth]\pgf@node@name.east) arc
                          (0:180:\insiderad-\pgflinewidth)--cycle;
      \fill[#2] (\pgf@node@name.base) ([xshift=\pgflinewidth]\pgf@node@name.west)  arc
                           (180:360:\insiderad-\pgflinewidth)--cycle;                    }}}}}  
 \makeatother
\tikzset{bdot/.style={circle,circle split,draw,circle split part fill={black,white},thin,inner sep=1pt}}%
\tikzset{wdot/.style={circle,circle split,draw,circle split part fill={white,black},thin,inner sep=1pt}}%
%
%
%
\newcommand\circlelinkpattern[2][]{
{
\pgfkeys{/linkpattern/.cd,#1}
\iflinkpatterntikzstarted\else%
\begin{tikzpicture}[/linkpattern/every linkpattern]%
\fi%
\iflinkpatterninverted%
\begin{scope}[yscale=-1]%
\fi%
\global\lpsize=\linkpatternsize
\edef\mylist{#2}
\foreach \x/\y in \mylist
{
\ifnum\x>\lpsize\global\lpsize=\x\fi
\ifnum\y>\lpsize\global\lpsize=\y\fi
}
%
\iflinkpatternaxis
\draw (0,0) circle (1);
\fi
\foreach\x in {1,...,\lpsize}
{
\pgfmathparse{(0.3*floor((\x-1)/\linkpatternfused)+0.7*((\x-0.5)/\linkpatternfused-0.5))*\linkpatternfused*360/\lpsize}
\coordinate[vertex] (v\x) at (\pgfmathresult:1);
}
\foreach \x/\y/\z in \mylist
{
\ifx\y\z%
\draw[edge] (v\x) .. controls ($0.5*(v\x)$) and  ($0.5*(v\y)$) .. (v\y);
\else
\draw[edge,decoration={markings,mark = at position 0.5 with { \arrow[semithick]{\z} }},postaction={decorate}] (v\x) .. controls ($0.5*(v\x)$) and  ($0.5*(v\y)$) .. (v\y);
\fi
}
\iflinkpatternnumbered%
\pgfmathparse{\lpsize/\linkpatternfused}
\global\lpsize=\pgfmathresult
\def\linkpatternnumbering{1,...,\lpsize}
\newdimen\angle
\foreach\x/\xx in \linkpatternnumbering
{
  \pgfmathsetmacro{\angle}{360/\lpsize*(\x-1)}
  \node[outer sep=1pt,anchor=180+\angle] at (\angle:1) {$\scriptstyle\xx$}; 
}
\fi%
\iflinkpatterninverted%
\end{scope}
\fi%
\iflinkpatterntikzstarted\else%
\end{tikzpicture}%
\fi%
}}%

We first define a link pattern of size $L$
to be a planar pairing of $\{1,\ldots,L\}$ viewed as boundary points
of a disk, e.g.,
\linkpatternnumberedtrue%
\[
\mathcal{L}_6=
\left\{
\linkpattern[shape=circle]{1/6,2/5,3/4},
\linkpattern[shape=circle]{1/6,2/3,4/5},
\linkpattern[shape=circle]{1/4,2/3,5/6},
\linkpattern[shape=circle]{1/2,3/6,4/5},
\linkpattern[shape=circle]{1/2,3/4,5/6}
\right\}
\]
We shall identify a link pattern with the fixed-point-free involution
which sends paired points to each other. Note that $L$ has to be even
for $\mathcal{L}_L$ to be nonempty.

The span of link patterns, $\mathbb{C}[\mathcal{L}_L]$,
can be identified with
a subspace of $(\mathbb{C}^2)^{\otimes L}$;
indeed, there is a linear map $\varphi$ which to each $\pi\in\mathcal{L}_L$ 
associates a  vector in $(\mathbb{C}^2)^{\otimes L}$,
which is given by the following explicit formula:
\begin{equation}\label{eq:spinloop}
\varphi(\pi)=\sum_{\epsilon\in\{0,1\}^L}
\prod_{1\le i<\pi(i)\le L} \delta_{\epsilon_i+\epsilon_{\pi(i)},1}
(-q)^{-\epsilon_i+\frac{1}{2}} \ketE{\epsilon}
\end{equation}
where on the r.h.s.\ the $\ketE{\epsilon}$ are the standard basis of
$(\mathbb{C}^2)^{\otimes L}$.
In appendix \ref{app:spinloop}, 
we show that $\varphi$ is injective,
and that for generic $q$ it is in fact an isomorphism from
$\mathbb{C}[\mathcal{L}_L]$ to the 
$U_1$-invariant subspace
of $(\mathbb{C}^2)^{\otimes L}$ (recall that $U_1$ is the horizontal subalgebra
of $U$).

The change of basis described by \eqref{eq:spinloop}
has the following diagrammatic
interpretation: one sums over all possible orientations of the arcs,
assigns a weight which is equal to $(-q)^{\frac{1}{2\pi}\text{total angle spanned by the arcs}}$,
and then reconstructs the spin state by recording the orientations
at the endpoints of the arcs,
identifying $\ketE{1}=\ketE{\uparrow}, \ketE{0}=\ketE{\downarrow}$.

The actual shape of the disk and the positioning of the points on its
boundary are irrelevant and result in gauge transformations (i.e.,
moving around the twist).
The choice above corresponds to all points aligned on a straight
region of the boundary, e.g.,
\begin{align*}
\linkpattern[shape=circle,numbered]{1/4,2/3}
&=
\linkpattern[boxed,numbered]{1/4,2/3}
\\
&=
\linkpattern[boxed,numbered]{1/4/>,2/3/>}
+
\linkpattern[boxed,numbered]{1/4/>,2/3/<}
+
\linkpattern[boxed,numbered]{1/4/<,2/3/>}
+
\linkpattern[boxed,numbered]{1/4/<,2/3/<}
\\
&=-q^{-1}\ketE{\uparrow\uparrow\downarrow\downarrow}+\ketE{\uparrow\downarrow\uparrow\downarrow}+\ketE{\downarrow\uparrow\downarrow\uparrow}-q\ketE{\downarrow\downarrow\uparrow\uparrow}
\end{align*}

The simplest way to compute the twist is to consider the
rotation  one-step to the left
of link patterns, $\tilde S$, and map it to the spin space.
We find that $\tilde S=(-q)^{a_{0,L}} S$ where $S$ is (untwisted)
spin rotation, compare with \eqref{eq:T1} at $k=0$.

For the sake of completeness we also indicate how to map covectors (and
therefore, operators). It is convenient to think of
covectors in the loop picture 
as planar pairings {\em outside}\/ the disk.\footnote{Note the asymmetry:
the region outside the disk is homeomorphic to an interval times a circle (``identified connectivities''),
not a disk (``unidentified connectivities''). 
It would be tempting to add the point at infinity to restore
the symmetry, but that is not allowed because we want the winding angle
of any closed curve to be $2\pi$.} The pairing between vectors and covectors
is to paste them together, and assign them $\tau$ to the power the number
of closed loops thus formed, where $\tau=-(q+q^{-1})$. Note that $\tau$
is nothing but $
\tikz[baseline=-3pt]{\draw[blue,very thick,decoration={markings,mark = at position 0.5 with { \arrow{>} }},postaction={decorate}] (0,0) circle (0.3cm);}
+
\tikz[baseline=-3pt]{\draw[blue,very thick,decoration={markings,mark = at position 0.5 with { \arrow{<} }},postaction={decorate}] (0,0) circle (0.3cm);}
$, where we have used the same rule as before that
the weight of a line is equal to $(-q)^{\frac{1}{2\pi}\text{total angle spanned by the arcs}}$. 
Therefore, the exact same graphical rule
must be used for covectors: sum over all orientations, give a weight of
$(-q)^{\frac{1}{2\pi}\text{total angle spanned by the arcs}}$
and build the dual spin state from the orientations
at the endpoints of the arcs,
with $\braE{1}=\braE{\uparrow}, \braE{0}=\braE{\downarrow}$,
$\braE{i}\ketE{j}=\delta_{ij}$.

\subsubsection{Fusion}
We shall not reproduce here the fusion procedure, i.e., consider
representation theory of $U$, but rather
consider only the non-affine part $U_1$.
We therefore view the $U_1$-module
$V=\mathbb{C}^{\ell+1}$ as a submodule of $(\mathbb{C}^2)^{\otimes\ell}$,
as well 
as the projection $p: (\mathbb{C}^2)^{\otimes\ell}\to V$
(see \ref{app:loopfused} for its explicit definition).

The $\ell$-fused link pattern of size $L$ of \cite{artic37}
is then obtained from ordinary link patterns of size $\ell L$ 
by grouping together vertices into groups indexed
by $r_i=1+\lfloor \frac{i-1}{\ell}\rfloor$
and by restricting to link patterns with no connections within a group, i.e.,
\[
\mathcal{L}_{\ell,L}=\{\pi\in\mathcal{L}_{\ell L}:
\pi(i)=j\ \Rightarrow r_i\ne r_j\}
\]
For example,
\[
\mathcal{L}_{3,4}
=
\left\{
\circlelinkpattern[fused=3]{1/6,2/5,3/4,7/12,8/11,9/10},
\circlelinkpattern[fused=3]{1/12,2/5,3/4,6/7,8/11,9/10},
\circlelinkpattern[fused=3]{1/12,2/11,3/4,5/8,6/7,9/10},
\circlelinkpattern[fused=3]{1/12,2/11,3/10,4/9,5/8,6/7}
\right\}
\]

In the fused case,
to a $\pi\in \mathcal{L}_{\ell,L}$ is then associated the corresponding vector 
{\em after projection}:
\begin{equation}\label{eq:projfus}
\varphi_\ell(\pi)
= \prod_{i=1}^L p_i\ \varphi(\pi)
\end{equation}
where $p_i$ is $p$ acting on the factors $\ell i+1,\ldots,\ell(i+1)$ of
the tensor product $(\mathbb{C}^2)^{\otimes\ell L}$.

Clearly, twist and fusion commute, so that we have the same twist
$(-1)^\ell K_1=(-q)^{a_0}$ in the fused case.

In appendix \ref{app:spinloop}, 
we show that the map $\varphi_\ell$
is well-defined and injective
provided $q^2$ is not a $r^{\text{th}}$ root of unity, $2\le r\le\ell$,
and that for generic $q$ it is in fact an isomorphism from
$\mathbb{C}[\mathcal{L}_{\ell;L}]$ to the 
$U_1$-invariant subspace
of $(\mathbb{C}^2)^{\otimes L}$.

The integrable inhomogeneous transfer matrix can be formulated directly
in the language of loops, as is done in \cite{artic37}.
Once converted to the spin space, it becomes exactly
the twisted transfer matrix $\mathbf{T}(z)$ discussed in \ref{sec:spinchain},
with twist given by $k=0$.

\subsubsection{Some results}
Consider now $\poly\Psi^{(0)}(z_1,\ldots,z_L)$. 
It is proportional to
$\Psi^{(0)}(z_1,\ldots,z_L)=\braC{0}\Phi(z_1)\ldots\Phi(z_L)\ketC{0}$,
and therefore, from the intertwining property of $\Phi(z)$,
is $U_1$-invariant. We conclude that it is in the image of $\varphi_\ell$ 
for generic $q$, and by continuity it is so for $q^2$ $(\ell+2)^{\text{th}}$ root
of unity.
We can therefore consider its preimage $\varphi_\ell^{-1}(\poly\Psi^{(0)})$.
By construction it is an eigenvector of the transfer matrix of the loop
model, with eigenvalue $1$. And it is a polynomial of degree $\ell n(n-1)$
in the $z_i$. At the particular value $q=-e^{-i\pi/(\ell+2)}$,
which is the only one considered in \cite{artic37}, $\poly\Psi^{(0)}$ coincides
with the eigenvector studied in that paper, and its
degree is the content of Conjecture 1 of \cite{artic37},
which is thus proved.

\rem{meaning of other twists for loop model?}

\subsection{Relation to supersymmetric fermions with exclusion}
\label{sec:susy}
In \cite{FNS-susy}, a supersymmetric
model of fermions on the one-dimensional lattice was introduced. It has
the exclusion rule that at most $\ell$ consecutive sites may be occupied.
In the case $\ell=1$, it was found to be equivalent to the XXZ spin
chain at $\Delta=-1/2$ \cite{YF-susy}. More generally, in a recent paper
\cite{Hag-susy}, Hagendorf conjectured that for any $\ell$,
this fermionic model was equivalent to the integrable spin $\ell/2$ chain
at $q=e^{i\pi/(\ell+2)}$. Note that models related by $q\to -q$ and $q\to 1/q$
are closely connected (in particular the choice $q\to -q$ depends
on the sign convention for the Hamiltonian). It was also stated that the
admissible twists are of the form $e^{2\pi i m/(\ell+2)a_0}$, where $m=0,1,\ldots,p$, $p=\ell+1$ for $\ell$ odd and $p=\ell/2$ for $\ell$ even. However,
it seems difficult to compare his twists with ours because the equivalence
is only conjectural and not explicit in the general case, and may modify the
twist. See also the computation of the Witten index in \cite[sect.~3.5]{Hag-susy}.

Adapting this discussion to our setting,
we conclude (conjecturally) that the Hamiltonian $\poly{H}$ at
$q=-e^{\pm i\pi/(\ell+2)}$ should be supersymmetric for all twists of
the form $\pm q^n$. Furthermore, it should have a supersymmetric ground state
exactly when $n\ne0\pmod {\ell+2}$.
This would imply that our state $\poly\Psi$, which has zero energy,
is indeed the (supersymmetric) ground state of $\poly H$.

\appendix

\section{Construction of \texorpdfstring{$q$-deformed}{q-deformed} parafermions}\label{app:parafermions}
Here we follow 
the procedure of Ding and Feigin~\cite{DF-para}, we intend to construct a level $\ell$ Drinfeld realization of $U_q(\widehat{\mathfrak{sl}(2)})$,
and to rewrite the current in terms of bosonic and parafermionic
operators (these being appropriate $q$-deformations of Conformal
Field Theory operators, cf the parafermionic fields of \cite{FZ-para}).

Let
\begin{equation}\label{eq:pf_series}
 \begin{aligned}
 e(z) &= \sum_{i \in \mathbb{Z}} e_i z^{-i-1} &
 f(z) &= \sum_{i \in \mathbb{Z}} f_i z^{-i-1}\\
 k^+(z) &= \sum_{i \in \mathbb{Z}^+_0} k^+_i z^{-i} &
 k^-(z) &= \sum_{i \in \mathbb{Z}^-_0} k^-_i z^{-i}
 \end{aligned}
\end{equation}

We want to construct operators that satisfy the following relations:
\begin{align}\label{eq:pf_relations}
\notag k^+_0 k^-_0 &= 1 = k^-_0 k^+_0\\
\notag k^{\pm} (z) k^{\pm} (w) &= k^{\pm} (w) k^{\pm} (z)\\
\notag k^- (z) k^+ (w) &= \frac{g\left(q^{-c} z/w\right)}{g\left(q^c z/w\right)} k^+ (w) k^- (z)\\
 k^{\pm}(z) e(w) &= g\left(q^{-\nicefrac{c}{2}} \left(w/z\right)^{\pm 1}\right)^{\mp 1} e(w) k^{\pm} (z)\\
\notag k^{\pm}(z) f(w) &= g\left(q^{\nicefrac{c}{2}} \left(w/z\right)^{\pm 1}\right)^{\pm 1} f(w) k^{\pm} (z)\\
\notag e(z) e(w) \left(z - q^2 w\right) &= e(w) e(z) \left(q^2 z - w\right) \\
\notag f(z) f(w) \left(q^2 z - w\right) &= f(w) f(z) \left(z - q^2 w\right) \\
\notag (q-q^{-1}) zw [e(z),f(w)] &= \delta \left( q^{-c} z/w \right) k^+ (q^{-\nicefrac{c}{2}} z) - \delta \left( q^c z/w \right) k^- (q^{\nicefrac{c}{2}} z)
\end{align}
where $c$ is a central element. 
The functions $g(z)$ and $\delta(z)$ are defined by:
\begin{align*}
 g(z) &= \frac{1-q^2 z}{q^2 - z} & \delta(z) &= \sum_{i \in \mathbb{Z}} z^i
\end{align*}

\begin{remark}
The Chevalley generators are given by:
 \begin{align*} 
 K_0 &= q^c k^-_0 & E_0 &=f_1 (k_0^+)^{-1} & F_0 &= k_0^+ e_{-1}\\
 K_1 &= k^+_0 & E_1 &=e_0 & F_1 &= f_0
 \end{align*}
\end{remark}

A Hopf algebra structure is obtained as follows. We define the co-multiplication:
\begin{equation}\label{eq:pf_coproduct}
\begin{aligned}
 \tilde\Delta(q^h) &= q^h \otimes q^h \\
 \tilde\Delta(e(z)) &= e(z) \otimes 1 + k^-(q^{\nicefrac{c_1}{2}}z) \otimes e(q^{c_1} z)\\
 \tilde\Delta(f(z)) &= f(q^{c_2} z) \otimes k^+(q^{\nicefrac{c_2}{2}}z) + 1 \otimes f(z)\\
 \tilde\Delta(k^{\pm} (z)) &= k^{\pm} (q^{\pm \nicefrac{c_2}{2}}z) \otimes k^{\pm} (q^{\mp \nicefrac{c_1}{2}} z)
\end{aligned}
\end{equation} 
where $c_1 = c \otimes 1$ and $c_2 = 1 \otimes c$;
the co-unit:
\begin{align}\label{eq:pf_counit}
\tilde\epsilon(q^h) &=1=\tilde\epsilon(k^{\pm} (z)) & \tilde\epsilon(e(z)) &= 0 = \tilde\epsilon(f(z)) 
\end{align}
and the antipode:
\begin{equation}\label{eq:pf_antipode}
 \begin{aligned}
\tilde a(q^h) &= q^{-h}\\
\tilde a(e(z)) &= - k^- (q^{-\nicefrac{c}{2}} z)^{-1} e(q^{-c} z)\\
\tilde a(f(z)) &= - f(q^{-c}) k^+(q^{-\nicefrac{c}{2}} z)^{-1}\\
\tilde a(k^{\pm} (z)) &= k^{\pm} (z)^{-1}
\end{aligned}
\end{equation}

\subsection{Highest weight modules}
We define a highest weight module as the unique module
$ V(\Lambda):=U_q(\widehat{\mathfrak{sl}(2)}) v_{\Lambda}$,
where $\Lambda = n_0 \Lambda_0 + n_1 \Lambda_1$ is the highest weight, 
and $v_\Lambda$ is the highest weight vector, which satisfies:
 \begin{align*}
 E_i v_{\Lambda} &= 0 & K_i v_{\Lambda} &= q^{(\Lambda,\alpha_i)} v_{\Lambda} & F_i^{(\Lambda,\alpha_i)+1} v_{\Lambda} &=0   
 \end{align*}
the inner product being defined by $(\Lambda_i,\alpha_j)=\delta_{i,j}$.
Except in the trivial case $\Lambda=0$, the modules thus generated are irreducible and infinite-dimensional. 

We will give a more explicit way of constructing these modules depending on
the eigenvalue $\ell = n_0+n_1$ of the central element $c$.

\subsection{Level \texorpdfstring{$1$}{1} realization}
We first build a realization of the $U_q (\widehat{\mathfrak{sl}(2)})$ algebra for the case $c=\ell=1$ in terms of bosonic currents.
We introduce a new set of bosonic operators, which satisfy the relations:
\begin{align*}
 [a_n,a_m] &= \delta_{n+m,0} \frac{[2n][n]}{n} & [a_n,\alpha]&=2\delta_{n,0}
\end{align*}

Then we can build the currents by:
\begin{equation} \label{eq:pf_level1}
\begin{aligned}
 e(z) &= e^{\alpha} z^{a_0} e^{\sum q^{-\nicefrac{n}{2}} \frac{a_{-n}}{[n]} z^n}e^{- \sum q^{-\nicefrac{n}{2}} \frac{a_n}{[n]} z^{-n}}\\
 f(z) &= e^{-\alpha} z^{-a_0} e^{-\sum q^{\nicefrac{n}{2}} \frac{a_{-n}}{[n]} z^n}e^{\sum q^{\nicefrac{n}{2}} \frac{a_n}{[n]} z^{-n}}\\
 k^+(z) &= q^{a_0} e^{(q-q^{-1})\sum a_{n} z^{-n}}\\
 k^-(z) &= q^{-a_0} e^{-(q-q^{-1})\sum a_{-n} z^{n}}
\end{aligned}
\end{equation}
these satisfy all the above relations.
We omit such computations.

They act on the usual bosonic Fock space $\Ho=
\bigoplus_{h\in \mathbb{Z}}\left\langle( \prod_{m>0}  a_{-m}^{k_m}\ketN{h}\right\rangle$
by the following rules:
\begin{align*}
 a_i f \ketN{h} &= a_i f \ketN{h} \quad\text{ if }i<0 & e^{\alpha} f \ketN{h} &= f \ketN{h+2}\\
 a_i f \ketN{h} &= [a_i,f] \ketN{h} \quad\text{ if }i>0 & a_0 f \ketN{h} &= f h \ketN{h}
\end{align*}
where $f \in \mathbb{C}[a_{-1},a_{-2},\ldots]$.

This produce two highest weight modules, depending on the parity of $h$, corresponding to the two highest weight vectors 
\begin{align*}
 v_{\Lambda_0} &= \ketN{0} & v_{\Lambda_1} &= \ketN{1}
\end{align*}


\subsection{Level \texorpdfstring{$\ell$}{l} realization}

Let us denote 
$\sigma_1:U\to \mathcal{L}(\Ho)$ 
the level $1$ realization.
We produce, using the coproduct $\tilde\Delta$,
a level $\ell$ representation as the tensor product
of $\ell$ copies of $V$, i.e.,
if we define $\tilde\Delta^{\ell} = (1\otimes \ldots \otimes \tilde\Delta)\circ \ldots \circ (1\otimes \tilde\Delta)\circ \tilde\Delta$, then
$
\sigma=\tilde\Delta^{\ell}(\sigma_1\otimes\cdots\otimes\sigma_1)
$ is the representation on $\HC=\Ho^{\otimes \ell}$.

We also define currents in $\HC$ in the obvious way:
\begin{equation}
\begin{aligned}
 \bm{e}(z) &= \tilde\Delta^{\ell} (e(z)) &
 \bm{f}(z) &= \tilde\Delta^{\ell} (f(z))\\
 \bm k^{\pm}(z) &= \tilde\Delta^{\ell} (k^{\pm}(z))& q^{\bm c} &= \tilde\Delta^{\ell} (q^c) = q^c \otimes \ldots \otimes q^c
\end{aligned}
\end{equation}
As $\tilde\Delta$ is a co-multiplication, this will satisfy automatically the relations~\eqref{eq:pf_relations}, with $\bm c = \ell$.

In that way we can compute $\bm e (z)$ for any $\ell$. For example, if $\ell=2$ we get:
\begin{equation}
  \bm{e}(z) = e(z) \otimes 1 + k^-(q^{\nicefrac{c_1}{2}} z) \otimes e(q^{c_1} z) = \bm{e}^1 (z) + \bm{e}^2 (z).
\end{equation}
where $c_i$ is the value of $c$ in $V_i$.

We perform a mode expansion of $\bm e(z)$, for generic $\ell$:
\begin{equation}
\bm{e}(z) = \sum_{i\leq \ell} \bm{e}^i (z)
\end{equation}
where (using $c_i=1$)
\begin{equation}
 \bm{e}^i (z) = \bigotimes_{j<i} k^{-}(q^{j-\nicefrac{1}{2} } z) \otimes e(q^{i-1} z) \bigotimes_{j>i} 1 
\end{equation}
where $j$ indicates the position on the tensor product.

We repeat the process for $\bm f(z)$:
\begin{equation}
 \bm{f}(z) = \sum_{i\leq \ell} \bm{f}^i (z)
\end{equation}
where (using $c_i=1$)
\begin{equation}
 \bm{f}^i (z) = \bigotimes_{j<i} 1 \otimes f(q^{\ell-i} z) \bigotimes_{j>i} k^+(q^{\ell-j+\nicefrac{1}{2}}z) 
\end{equation}
The same goes for $\bm k^{\pm}(z)$:
\begin{equation}
 \bm k^{\pm} (z) = \bigotimes_j k^{\pm} (q^{\pm \frac{\ell-2j +1}{2}} z) 
\end{equation}

Let us denote $\ketC{k_1,\ldots,k_{\ell}}=\ketN{k_1}\otimes \ldots\otimes \ketN{k_\ell}$, where $k_1 = 0,\ 1$.
Let $k = \sum_i k_i$. 
$\ketC{k_1,\ldots,k_{\ell}}$,
as a tensor product of highest weight vectors, is a highest weight vector 
with highest weight $\Lambda = (\ell-k) \Lambda_0 + k \Lambda_1$.
Note that for a given $\Lambda$, except for $k=0$ or $\ell$, this construction is not unique. For our purposes, any such highest weight vector 
provides a realization of the highest weight module,
and we will use the notation $\ketC{k}$ for any of them 
(our results do not depend on the choice, i.e., when we permute $k_i \leftrightarrow k_j$).
Moreover, these vectors are normalized such that $\braket{k}{k} = 1$.

\subsection{Parafermions}
The computation of commutation relations between $\bm{e}^i (z)$ (or $\bm{f}^i (z)$) and $\bm k^{\pm} (w)$ it is not complicated, although it is long.
The result for $\bm{e}^i(z)$ is:
\begin{equation}\label{eq:pf_com_e_h}
\begin{aligned}
 \bm{e}^i (z) \bm k^+ (w) &= :\bm{e}^i (z) \bm k^+ (w): \\
 \bm k^+ (w) \bm{e}^i (z) &= \frac{q^2 w - q^{-\nicefrac{\ell}{2}}z}{w - q^{2- \nicefrac{\ell}{2}}z }:\bm{e}^i (z) \bm k^+ (w): \\
 \bm{e}^i (z) \bm k^- (w) &= \frac{z - q^{-2-\nicefrac{\ell}{2}}w}{z - q^{2- \nicefrac{\ell}{2}}w }:\bm{e}^i (z) \bm k^- (w): \\
 \bm k^- (w) \bm{e}^i (z) &= q^{-2} :\bm{e}^i (z) \bm k^- (w):
\end{aligned}
\end{equation}
and for $\bm{f}^i(z)$ is:
\begin{equation}\label{eq:pf_com_f_h}
\begin{aligned}
 \bm{f}^i (z) \bm k^+ (w) &= :\bm{f}^i (z) \bm k^+ (w): \\
 \bm k^+ (w) \bm{f}^i (z) &= \frac{w - q^{2+\nicefrac{\ell}{2}}z}{q^2 w - q^{ \nicefrac{\ell}{2}}z }:\bm{f}^i (z) \bm k^+ (w): \\
 \bm{f}^i (z) \bm k^- (w) &= \frac{z - q^{2+\nicefrac{\ell}{2}}w}{z - q^{-2+ \nicefrac{\ell}{2}}w }:\bm{f}^i (z) \bm k^- (w): \\
 \bm k^- (w) \bm{f}^i (z) &= q^{-2} :\bm{f}^i (z) \bm k^- (w):
\end{aligned}
\end{equation}

Surprisingly, these results do not depend on $i$. 
This fact, they tell us that we can factor $\bm e^i (z)$ and $\bm f^i (z)$ into two parts: one (parafermionic part) that commutes with $\bm k^{\pm} (w)$ and depends on $i$ and a second part (bosonic part) that does not depend on $i$ and does not commute with $\bm k^{\pm} (w)$.

Following Ding and Feigin, let us build the bosonic part.
Define (let $n$ be a non-negative integer):
\begin{equation}
\begin{aligned}
  \bm{a}_{\pm n} &= \sum_i \bm{a}_{\pm n}^i \\
  \bm{a}_{\pm n}^i &= \bigotimes_{j<i} 1 \otimes q^{-n\frac{\ell-2i+1}{2}} a_{\pm n} \bigotimes_{j>i} 1
\end{aligned}
\end{equation}
(correcting a sign misprint in \cite{DF-para}).

We could think that $\bm a_{n}$ is only the result of applying the co-multiplication to $a_n$.
But this is false. 
In fact, there is no Hopf algebra for these operators.

These new operators satisfy the following commutation realation:
\begin{align}
 [\bm{a}_n,\bm{a}_m] &= \delta_{n+m,0}\frac{[2n][\ell n]}{n}
\end{align}
It can be easily proved.

The main reason for this definition is the fact that:
\begin{equation}
 \bm k^{\pm} (z) = q^{\pm \bm{a}_0} e^{\pm (q-q^{-1}) \sum_n \bm{a}_{\pm n} z^{\mp n}}
\end{equation}
where, for example, $q^{\bm{a}_0}$ means $q^{a_0} \otimes \ldots \otimes q^{a_0}$.
Notice that this is the natural generalization of $k^{\pm} (z)$.

Now the idea is to factor $\bm e^i (z)$ and $\bm f^i (z)$ using these new bosonic operators.
That is, we want to write them as $\bm{e}^i (z)= \mathfrak{e}(z) \xi^{+i} (z)$ and $\bm{f}^i (z) = \mathfrak{f}(z) \xi^{-i} (z)$, where $\mathfrak{e}(z)$ and $\mathfrak{f}(z)$ are expressed in terms of $\bm{a}_{\pm n}$ and $e^{\bm \alpha} = \bigotimes_j e^{\alpha}$
(compared to the notations of the main text, we have renormalized the
zero mode as $\bm\alpha=2\ell\bm\beta$).

The following Ansatz is made:
\begin{equation}
\begin{aligned}
 \mathfrak{e}(z) &= e^{\nicefrac{\bm{\alpha}}{\ell}} z^{\nicefrac{\bm{a}_0}{\ell}} e^{\sum_n q^{-\frac{n\ell}{2}} \frac{\bm{a}_{-n}}{[n \ell]} z^n} e^{-\sum_n q^{-\frac{n\ell}{2}} \frac{\bm{a}_n}{[n \ell]} z^{-n}}\\
 \mathfrak{f}(z) &= e^{-\nicefrac{\bm{\alpha}}{\ell}} z^{-\nicefrac{\bm{a}_0}{\ell}} e^{-\sum_n q^{\frac{n\ell}{2}} \frac{\bm{a}_{-n}}{[n \ell]} z^n} e^{\sum_n q^{\frac{n\ell}{2}} \frac{\bm{a}_n}{[n \ell]} z^{-n}}
\end{aligned}
\end{equation}

In this way, the commutation relations of $\mathfrak{e}(z)$ (or $\mathfrak{f}(z)$) and $\bm k^{\pm} (w)$ have the form of~\eqref{eq:pf_com_e_h} and~\eqref{eq:pf_com_f_h}, i.e.,
\begin{equation}
\begin{aligned}
0&=[\xi^{+i} (z), \bm k^{\pm} (w)] \\
0&=[\xi^{-i} (z), \bm k^{\pm} (w)]
\end{aligned}
\end{equation}
for all $i$.

In fact, one can prove that 
 \begin{align}
  [\bm{a}_{n}, \xi^{\pm i}(w)] &=0 &
  [\bm{\alpha}, \xi^{\pm i} (w)] &=0
 \end{align}
for any integer $n$ and $i \in \{1,2,\ldots, \ell\}$.

The operators $\xi^+(z) = \sum_i \xi^{+i}(z)$ and $\xi^- (z) = \sum_i \xi^{-i} (z)$ are the so-called parafermionic operators.

We can split a vector $\ketC{k}$ into a bosonic component and a parafermionic component. 
The first being an average $\ketN{k} = \ketN{\nicefrac{k}{\ell}} \otimes \ldots \otimes \ketN{\nicefrac{k}{\ell}} $.
This requires that we enlarge the Fock space to $\Ho = \bigoplus_{h \in \nicefrac{\mathbb{Z}}{\ell}} \left\langle (\prod_{m>0} a_{-m}^{k_m} \ketN{h} \right\rangle$. 

\subsection{The wheel condition}
These operators satisfy an important relation.
\begin{theorem}\label{thm:pf-wheel}
$\bm e(z), \bm f(z)$ satisfy the wheel condition:
\begin{align*}
 \bm{e}(z) \bm{e}(q^2 z) \ldots \bm{e} (q^{2\ell}z) &= 0\\
 \bm{f}(q^{2\ell} z) \ldots \bm{f}(q^2 z) \bm{f}(z) &= 0
\end{align*}
\end{theorem}
This is theorem 2.5 in~\cite{DF-para}.

In order to prove this fact, we need the folowing relations for $\bm{e}^i (z)$ (let $j>i$):
\begin{equation}
\begin{aligned}
 \bm{e}^i (z) \bm{e}^j (w) &= \frac{z-q^{-2} w}{z-q^2 w} :\bm{e}^i (z) \bm{e}^j (w):\\
 \bm{e}^i (z) \bm{e}^i (w) &= q^{2i-2} (z-w)(z-q^{-2} w) :\bm{e}^i (z) \bm{e}^i (w):\\
 \bm{e}^j (z) \bm{e}^i (w) &= q^{-2}:\bm{e}^j (z) \bm{e}^i (w):
\end{aligned}
\end{equation}
And the equivalent for $\bm{f}^i(z)$:
\begin{equation}\label{eq:pf_f_relations}
\begin{aligned}
 \bm{f}^i (z) \bm{f}^j (w) &= \frac{z-q^2 w}{q^2 z - w} :\bm{f}^i (z) \bm{f}^j (w):\\
 \bm{f}^i (z) \bm{f}^i (w) &= q^{2\ell-2i}(z-w)(z-q^2 w) :\bm{f}^i (z) \bm{f}^i (w):\\
 \bm{f}^j (z) \bm{f}^i (w) &= :\bm{f}^j (z) \bm{f}^i (w):
\end{aligned}
\end{equation}

Note that when $w=q^2 z$ (or $w=q^{-2} z$ for the $\bm{f}^i(z)$ case) these expressions simplify. This is our main tool to prove theorem~\ref{thm:pf-wheel}:

\begin{proof}
The product $\bm{e}(z_1) \ldots \bm{e}(z_n)$ can be written as the sum over all possible decompositions:
\[
 \bm{e}(z_1) \ldots \bm{e}(z_n) = \sum_{\epsilon_1,\ldots,\epsilon_n} \bm{e}^{\epsilon_1} (z_1) \ldots \bm{e}^{\epsilon_n} (z_n)
\]

The regular part, that is the $:\ldots :$ part, has no poles (aside $0$ and $\infty$).
Thus, setting $w=q^2 z$, we see that only $\bm{e}^j (z) \bm{e}^i (q^2 z)$ survives (we are assuming $j>i$).
\[
 \bm{e}(z) \ldots \bm{e}(q^{2n-2}z) = \sum_{\epsilon_1>\ldots>\epsilon_n} \bm{e}^{\epsilon_1} (z) \ldots \bm{e}^{\epsilon_n} (q^{2n-2}z)
\]
but such sequence $\epsilon_1 > \ldots > \epsilon_n$ is impossible if $n=\ell+1$, considering that $\epsilon_i \in \{1,\ldots,\ell\}$.
The result follows.
The proof for $\bm f^i (z)$ is identical.
\end{proof}

\section{Graphical calculus} 
We describe here the graphical calculus to represent
$U_q(\widehat{\mathfrak{sl}(2)})$ invariants. It is a convenient way
to derive various relations while keeping track of the spaces involved,
i.e., it acts in a similar way as a ``type checking'' tool for programming.

\renewcommand\ss{\scriptstyle}
\subsection{Basics}
Vector spaces are represented by oriented lines, such that tensor products
correspond to stacking lines from left to right if looking in the direction
of the orientation. Highest weight representations $\HC$ 
are denoted by thick lines,
evaluation representations $V_z$ by thin lines with a label $z$. 
Thus, a (type I) vertex operator
$\Phi(z)$ becomes:
\[
\Phi(z)=
\begin{tikzpicture}[baseline=-3pt,scale=0.75,rotate=180]
\draw[ultra thick,decoration={markings, mark = at position 0.25 with {\arrow{>}}},postaction={decorate}] (0,0) -- (2,0);
\draw[decoration={markings, mark = at position 0.5 with {\arrow{>}}},postaction={decorate}] (1,0) -- node[left] {$\ss z$} (1,-1);
\end{tikzpicture}
\]
and its (vacuum to vacuum) correlation function is: 
\[
\Psi(z_1,\ldots,z_L)=\braC{0}\Phi(z_1)\ldots\Phi(z_L)\ketC{0}
=
\begin{tikzpicture}[baseline=-3pt,scale=0.75,rotate=180]
\draw[ultra thick,decoration={markings, mark = at position 0.2 with {\arrow{>}}},postaction={decorate}] (0,0) node[thin,circle,draw,fill=white,inner sep=2pt] {} -- (2,0) (3,0) -- (5,0) node[thin,circle,draw,fill=white,inner sep=2pt] {};
\draw[ultra thick,dotted] (2,0) -- (3,0);
\draw[decoration={markings, mark = at position 0.5 with {\arrow{>}}},postaction={decorate}] (1,0) -- node[left] {$\ss z_L$} (1,-1);
\draw[decoration={markings, mark = at position 0.5 with {\arrow{>}}},postaction={decorate}] (4,0) -- node[left] {$\ss z_1$} (4,-1);
\end{tikzpicture}
\]
where the white dots on the left (resp.\ right) correspond to $\braC{0}$ (resp.\ $\ketC{0}$).

We need two more objects obtained by applying representations to the universal 
$R$-matrix $\mathcal{R}$ (or to $\mathcal P(\mathcal{R}^{-1})$): 
the $R$-matrix
\[
R_+(z_1/z_2)=
\begin{tikzpicture}[baseline=-3pt,scale=0.75,shift={(0,0.5)}]
\draw[decoration={markings, mark = at position 0.25 with {\arrow{<}}},postaction={decorate}] (1,0) -- node[right] {$\ss z_1$} (0.6,-0.4) (0.4,-0.6) -- (0,-1);
\draw[decoration={markings, mark = at position 0.25 with {\arrow{<}}},postaction={decorate}] (0,0) -- node[pos=0.25,left] {$\ss z_2$} (1,-1);
\end{tikzpicture}
\qquad
R_-(z_1/z_2)=
\begin{tikzpicture}[baseline=-3pt,scale=0.75,shift={(0,0.5)}]
\draw[decoration={markings, mark = at position 0.25 with {\arrow{<}}},postaction={decorate}] (-1,0) -- node[left] {$\ss z_2$} (-0.6,-0.4) (-0.4,-0.6) -- (0,-1);
\draw[decoration={markings, mark = at position 0.25 with {\arrow{<}}},postaction={decorate}] (0,0) -- node[pos=0.25,right] {$\ss z_1$} (-1,-1);
\end{tikzpicture}
\]
where the argument is the ratio of spectral parameters of left and right incoming lines (note that $R_+$ and $R_-$ only differ by a normalization);
and the $L$-matrix
\begin{align*}
L_+(z)&=\begin{tikzpicture}[baseline=-3pt,scale=0.75,rotate=180]
\draw[ultra thick,decoration={markings, mark = at position 0.25 with {\arrow{>}}},postaction={decorate}] (-1,0) -- (1,0);
\draw[decoration={markings, mark = at position 0.25 with {\arrow{>}}},postaction={decorate}] (0,1) -- node[right] {$\ss q^{\ell/2}z$} (0,0.2) (0,-0.2) -- node[right] {$\ss q^{-\ell/2}z$} (0,-1);
\end{tikzpicture}
&
a(L_+(z))&=\begin{tikzpicture}[baseline=-3pt,scale=0.75,rotate=180]
\draw[ultra thick,decoration={markings, mark = at position 0.25 with {\arrow{>}}},postaction={decorate}] (1,0) -- (-1,0);
\draw[decoration={markings, mark = at position 0.25 with {\arrow{>}}},postaction={decorate}] (0,1) -- node[right] {$\ss q^{-\ell/2}z$} (0,0.2) (0,-0.2) -- node[right] {$\ss q^{\ell/2}z$} (0,-1);
\end{tikzpicture}
\\
L_-(z)&=\begin{tikzpicture}[baseline=-3pt,scale=0.75,rotate=180]
\draw[ultra thick,decoration={markings, mark = at position 0.25 with {\arrow{>}}},postaction={decorate}] (-1,0) -- (-0.2,0) (0.2,0) -- (1,0);
\draw[decoration={markings, mark = at position 0.25 with {\arrow{>}}},postaction={decorate}] (0,1) -- node[right] {$\ss q^{-\ell/2}z$} (0,0) -- node[right] {$\ss q^{\ell/2}z$} (0,-1);
\end{tikzpicture}
&
a(L_-(z))&=\begin{tikzpicture}[baseline=-3pt,scale=0.75,rotate=180]
\draw[ultra thick,decoration={markings, mark = at position 0.25 with {\arrow{>}}},postaction={decorate}] (1,0) -- (0.2,0) (-0.2,0) -- (-1,0);
\draw[decoration={markings, mark = at position 0.25 with {\arrow{>}}},postaction={decorate}] (0,1) -- node[right] {$\ss q^{\ell/2}z$} (0,0) -- node[right] {$\ss q^{-\ell/2}z$} (0,-1);
\end{tikzpicture}
\end{align*}
Note that thin lines pick up a factor of $q^{\pm\ell}$ when they cross thick lines.

Lines can also be slid across intersections
of other lines (Reidemeister move III) due to the Yang--Baxter equation, or 
across the trivalent vertex of a VO ($L$/VO commutation).

Finally, 
it is easy to check on generators that the square of the antipode
satisfies $a^2(x)=q^{-a_0+4d}xq^{a_0-4d}$ for all $x$ in $U$. This implies that any line that does a full $\pm 2\pi$ rotation
can be replaced with a straight line with an insertion of $q^{\pm(a_0-4d)}$ times
a central element, e.g.,
\begin{equation}\label{eq:doubledual}
\begin{tikzpicture}[baseline=-3pt,scale=0.75,rotate=180]
\draw[decoration={markings, mark = at position 0.15 with {\arrow{<}}},postaction={decorate}]
(0,-1) -- (0,-0.2) (0,0) arc (0:-310:0.2) (0,0) -- (0,1);
\end{tikzpicture}
=c\ 
\begin{tikzpicture}[baseline=-3pt,scale=0.75,rotate=180]
\draw[decoration={markings, mark = at position 0.25 with {\arrow{<}}},postaction={decorate}] 
(0,-1) -- (0,0) node[circle,fill=black,inner sep=0.05cm] {} -- (0,1);
\end{tikzpicture}
\qquad
\begin{tikzpicture}[baseline=-3pt,scale=0.75,rotate=-90]
\draw[ultra thick,decoration={markings, mark = at position 0.15 with {\arrow{<}}},postaction={decorate}]
(0,-1) -- (0,-0.2) (0,0) arc (0:-310:0.2) (0,0) -- (0,1);
\end{tikzpicture}
=C\ 
\begin{tikzpicture}[baseline=-3pt,scale=0.75,rotate=-90]
\draw[ultra thick,decoration={markings, mark = at position 0.25 with {\arrow{<}}},postaction={decorate}] 
(0,-1) -- (0,0) node[circle,fill=black,inner sep=0.05cm] {} -- (0,1);
\end{tikzpicture}
\end{equation}
where the dot represents insertion of $q^{a_0-4d}$.
We shall not need the values $c,C$ of the central element in what follows.
The $q^{-4d}$ acts as a multiplicative shift of $q^{-4}$ of the spectral
parameter, which propagates along the lines.

\subsection{Proof of \texorpdfstring{$q$KZ}{qKZ} equation}\label{app:qKZ}
We begin with a local relation: consider the following equality,
obvious diagrammatically:
\[
\begin{tikzpicture}[baseline=-3pt,scale=0.75,rounded corners=0.3cm,rotate=180]
\draw[ultra thick,decoration={markings, mark = at position 0.45 with {\arrow{>}}},postaction={decorate}] (0,0) -- (2.8,0) (3.2,0) -- (4,0);
\draw[decoration={markings, mark = at position 0.2 with {\arrow{>}}},postaction={decorate}] (2,0) -- (2,-0.8) -- node[above] {$\ss z$} (3,-0.8) -- (3,1) [sharp corners] -- (2.1,1) arc (270:-60:0.2cm) [rounded corners=0.3cm] (1.9,1) -- (1,1) -- node[right] {$\ss q^{-\ell} z$} (1,0.1) (1,-0.1)
-- node[right] {$\ss q^{-2\ell}z$} (1,-1);
\end{tikzpicture}
=
\begin{tikzpicture}[baseline=-1.5cm,scale=0.5,rotate=180]
\draw[rounded corners=0.2cm,ultra thick,decoration={markings, mark = at position 0.25 with {\arrow{<}}},postaction={decorate}] (3,5) -- (-2,5) -- (-2,2) -- (-0.2,2) (0.2,2) -- (1,2) -- (1,3) -- (-1,3) -- (-1,4) -- (2,4) -- (2,1) -- (-3,1);
\draw[decoration={markings, mark = at position 0.9 with {\arrow{>}}},postaction={decorate}] (0,3) -- (0,1.6) arc (0:-300:0.17cm) (0,1.4) -- (0,1.2) (0,0.8) -- node[right] {$\ss q^{-2\ell}z$} (0,0);
\end{tikzpicture}
=
\begin{tikzpicture}[baseline=-1.35cm,scale=0.75,rotate=180]
\draw[ultra thick,decoration={markings, mark = at position 0.57 with {\arrow{>}}},postaction={decorate}] (-2,2) -- (-1.1,2) arc (-90:220:0.2cm) (-0.8,2) -- (0.9,2) arc (90:-220:0.2cm) (1.1,2) -- (2,2);
\draw[decoration={markings, mark = at position 0.9 with {\arrow{>}}},postaction={decorate}] (0,2) -- (0,1.1) arc (0:-320:0.2cm) (0,0.9) -- node[right] {$\ss q^{-2\ell}z$} (0,0); 
\end{tikzpicture}
\]
If we use \eqref{eq:doubledual}, 
recognizing conjugation by $q^{a_0-4d}$ on the thick line,
the coproduct $\Delta(q^{a_0-4d})=q^{a_0-4d}\otimes q^{a_0-4d}$ and finally the intertwining property of the VO, we can simplify this to:
\[
\begin{tikzpicture}[baseline=-3pt,scale=0.75,rounded corners=0.3cm,rotate=180]
\draw[ultra thick,decoration={markings, mark = at position 0.45 with {\arrow{>}}},postaction={decorate}] (0,0) -- (2.8,0) (3.2,0) -- (4,0);
\draw[decoration={markings, mark = at position 0.2 with {\arrow{>}}},postaction={decorate}] (2,0) -- (2,-0.8) -- node[above] {$\ss z$} (3,-0.8) -- (3,1)  -- (2,1) node[circle,fill=black,inner sep=0.05cm] {} -- (1,1) -- (1,0.1) (1,-0.1)
-- node[right] {$\ss sz$} (1,-1);
\end{tikzpicture}
=
\begin{tikzpicture}[baseline=-3pt,scale=0.75,rotate=180]
\draw[ultra thick,decoration={markings, mark = at position 0.25 with {\arrow{>}}},postaction={decorate}] (0,0) -- (2,0);
\draw[decoration={markings, mark = at position 0.5 with {\arrow{>}}},postaction={decorate}] (1,0) -- node[right] {$\ss sz$} (1,-1);
\end{tikzpicture}
\]
where $s=q^{-2(\ell+2)}$, and we recall that the dot represents insertion of $K_1=q^{a_0}$ accompanied by a shift of the spectral parameter $z\to q^{-4}z$, which we pulled out of the thin line. \rem{plus a power of $q$, like $q^{\ell/4}$, which
probably compensates when we hit the bra/kets}

Now apply this relation to the correlation function $\Psi(z_1,\ldots,z_L)$
described graphically above, and use the highest/lowest weight property of
$\ketC{0}$ and $\braC{0}$, i.e., $L_+(z)\ketC{0}=1$ and $\braC{0}L_-(z)=1$:
\begin{equation}\label{eq:qKZgraph}
\begin{tikzpicture}[baseline=-3pt,scale=0.75,rotate=180]
\draw[ultra thick,decoration={markings, mark = at position 0.25 with {\arrow{>}}},postaction={decorate}] (0,0) node[thin,circle,draw,fill=white,inner sep=2pt] {} -- (2,0) (3,0) -- (5,0) (6,0) -- (8,0) node[thin,circle,draw,fill=white,inner sep=2pt] {};
\draw[ultra thick,dotted] (2,0) -- (3,0) (5,0) -- (6,0);
\draw[decoration={markings, mark = at position 0.5 with {\arrow{>}}},postaction={decorate}] (1,0) -- node[left] {$\ss z_L$} (1,-1);
\draw[decoration={markings, mark = at position 0.5 with {\arrow{>}}},postaction={decorate}] (4,0) -- node[left] {$\ss s z_i$} (4,-1);
\draw[decoration={markings, mark = at position 0.5 with {\arrow{>}}},postaction={decorate}] (7,0) -- node[left] {$\ss z_1$} (7,-1);
\end{tikzpicture}
=
\begin{tikzpicture}[baseline=-3pt,scale=0.75,rounded corners,rotate=180]
\draw[ultra thick,decoration={markings, mark = at position 0.25 with {\arrow{>}}},postaction={decorate}] (0,0) node[thin,circle,draw,fill=white,inner sep=2pt] {} -- (2,0) (3,0) -- (5,0) (6,0) -- (8,0) node[thin,circle,draw,fill=white,inner sep=2pt] {};
\draw[ultra thick,dotted] (2,0) -- (3,0) (5,0) -- (6,0);
\draw[decoration={markings, mark = at position 0.75 with {\arrow{>}}},postaction={decorate}] (1,0) -- node[left,pos=0.8] {$\ss z_L$} (1,-1);
\draw[decoration={markings, mark = at position 0.9 with {\arrow{<}}},postaction={decorate}] (3.75,-1) -- (3.75,-0.5) -- node[above] {$\ss s z_i$} (1.1,-0.5) (0.9,-0.5) -- (-0.5,-0.5) -- (-0.5,0.5) -- node[circle,fill=black,inner sep=0.05cm] {} (8.5,0.5) -- (8.5,-0.5) -- node[above,pos=0.75] {$\ss z_i$} (4.25,-0.5) -- (4.25,0);
\draw[decoration={markings, mark = at position 0.75 with {\arrow{>}}},postaction={decorate}] (7,0) -- (7,-0.4) (7,-0.6) -- node[left] {$\ss z_1$} (7,-1);
\end{tikzpicture}
\end{equation}
which is exactly the quantum Knizhnik--Zamolodchikov equation \eqref{eq:qKZ}.

In case $\braC{0}\ldots\ketC{0}$ 
is replaced with
$\braC{k}\ldots\ketC{k}$, $k=0,\ldots,\ell$,
there is an additional contribution
coming from the zero modes of $L_\pm(z)\sim q^{a_0\otimes a_0/2}$, so that
the twist in the dot becomes $q^{(1+k)a_0}$.

The $q$KZ equation is valid for the correlation function of a product of
arbitrary VOs; however, in the special case of perfect VOs, we have the additional relation \eqref{eq:exchPhi}:
\[
\begin{tikzpicture}[baseline=-3pt,scale=0.75,rounded corners=0.2cm,rotate=180]
\draw[ultra thick,decoration={markings, mark = at position 0.5 with {\arrow{>}}},postaction={decorate}] (0,0) -- (3,0);
\draw[decoration={markings, mark = at position 0.5 with {\arrow{>}}},postaction={decorate}]  (1,0) -- node[right] {$\ss z_2$} (1,-1) -- (2,-2);
\draw[decoration={markings, mark = at position 0.5 with {\arrow{>}}},postaction={decorate}]  (2,0) -- node[left] {$\ss z_1$} (2,-1) -- (1,-2);
\end{tikzpicture}
=
\begin{tikzpicture}[baseline=-3pt,scale=0.75,rotate=180]
\draw[ultra thick,decoration={markings, mark = at position 0.5 with {\arrow{>}}},postaction={decorate}]  (0,0) -- (3,0);
\draw[decoration={markings, mark = at position 0.5 with {\arrow{>}}},postaction={decorate}]  (1,0) -- node[right] {$\ss z_1$} (1,-2);
\draw[decoration={markings, mark = at position 0.5 with {\arrow{>}}},postaction={decorate}]  (2,0) -- node[left] {$\ss z_2$} (2,-2);
\end{tikzpicture}
\]
where the flat crossing corresponds to another normalization of the $R$-matrix,
denoted by $R(z)$ in the text.
This in turn implies the exchange relation  \eqref{eq:exch}
for correlation functions, and, combined with the $q$KZ equation at $i=L$,
produces the cyclicity property:
\[
\begin{tikzpicture}[baseline=-3pt,scale=0.75,rotate=180]
\draw[ultra thick,decoration={markings, mark = at position 0.25 with {\arrow{>}}},postaction={decorate}] (0,0) node[thin,circle,draw,fill=white,inner sep=2pt] {} -- (3,0) (4,0) -- (7,0) node[thin,circle,draw,fill=white,inner sep=2pt] {};
\draw[ultra thick,dotted] (3,0) -- (4,0);
\draw[decoration={markings, mark = at position 0.5 with {\arrow{>}}},postaction={decorate}] (1,0) -- node[right] {$\ss s z_L$} (1,-1);
\draw[decoration={markings, mark = at position 0.5 with {\arrow{>}}},postaction={decorate}] (2,0) -- node[left] {$\ss z_{L-1}$} (2,-1);
\draw[decoration={markings, mark = at position 0.5 with {\arrow{>}}},postaction={decorate}] (6,0) -- node[left] {$\ss z_1$} (6,-1);
\end{tikzpicture}
\propto
\begin{tikzpicture}[baseline=-3pt,scale=0.75,rotate=180]
\draw[ultra thick,decoration={markings, mark = at position 0.25 with {\arrow{>}}},postaction={decorate}] (0,0) node[thin,circle,draw,fill=white,inner sep=2pt] {} -- (3,0) (4,0) -- (7,0) node[thin,circle,draw,fill=white,inner sep=2pt] {};
\draw[ultra thick,dotted] (3,0) -- (4,0);
\draw[decoration={markings, mark = at position 0.5 with {\arrow{>}}},postaction={decorate}] (2,0) -- node[left] {$\ss z_{L-1}$} (2,-1);
\draw[decoration={markings, mark = at position 0.5 with {\arrow{>}}},postaction={decorate}] (5,0) -- node[left] {$\ss z_1$} (5,-1);
\draw[decoration={markings, mark = at position 0.5 with {\arrow{<}}},postaction={decorate},rounded corners] (1,-1) -- (1,-0.5) -- (-0.5,-0.5) -- (-0.5,0.5) -- node[circle,fill=black,inner sep=0.05cm] {} (7.5,0.5) -- (7.5,-0.5) -- node[above] {$\ss z_L$} (6,-0.5) -- (6,0);
\end{tikzpicture}
\]
The change from $\Psi$ to $\poly\Psi$ removes the proportionality factor,
which leads to \eqref{eq:cycl}.

\subsection{Proof of eigenvector property}\label{sec:grapheig}
We provide here a graphical proof of the eigenvector property
of section \ref{sec:eig}. What follows 
is not rigorous because the VO construction
becomes divergent when $|q|=1$; however, the proper proof, given in the text, is
along the exact same lines as this graphical proof, except with diverging
prefactors removed. We shall therefore ignore prefactors in what follows.

Recall the dual vertex operator \eqref{eq:dualvo}. The properties
\eqref{eq:dualrel} it satisfies can be depicted as:
\begin{subequations}
\begin{align}\label{eq:dualrelgraph}
\Phi_b(z)\Phi^\ast_{b'}(z)&=\delta_{b,b'} g_k
&
\begin{tikzpicture}[baseline=-3pt,scale=0.75,rounded corners,rotate=180]
\draw[ultra thick,decoration={markings, mark = at position 0.25 with {\arrow{>}}},postaction={decorate}] (0,0) -- (3,0);
\draw[decoration={markings, mark = at position 0.5 with {\arrow{<}}},postaction={decorate}] (1,0) -- node[left] {$\ss z$} (1,-1);
\draw[decoration={markings, mark = at position 0.5 with {\arrow{>}}},postaction={decorate}] (2,0) -- node[left] {$\ss z$} (2,-1);
\end{tikzpicture}
&=g_k\ 
\begin{tikzpicture}[baseline=-3pt,scale=0.75,rounded corners,rotate=180]
\draw[ultra thick,decoration={markings, mark = at position 0.25 with {\arrow{>}}},postaction={decorate}] (0,0) -- (3,0);
\draw[decoration={markings, mark = at position 0.5 with {\arrow{>}}},postaction={decorate}] (1,-1) -- (1,-0.3) -- node[above] {$\ss z$} (2,-0.3) -- (2,-1);\end{tikzpicture}
\\
\sum_{b=0}^\ell \Phi^\ast_b(z)\Phi_b(z)&= g_k
&
\begin{tikzpicture}[baseline=-3pt,scale=0.75,rounded corners,rotate=180]
\draw[ultra thick,decoration={markings, mark = at position 0.25 with {\arrow{>}}},postaction={decorate}] (0,0) -- (3,0);
\draw[decoration={markings, mark = at position 0.5 with {\arrow{>}}},postaction={decorate}] (1,0) -- (1,-1) -- node[above] {$\ss z$} (2,-1) -- (2,0);
\end{tikzpicture}
&=g_k\ 
\begin{tikzpicture}[baseline=-3pt,scale=0.75,rounded corners,rotate=180]
\draw[ultra thick,decoration={markings, mark = at position 0.25 with {\arrow{>}}},postaction={decorate}] (0,0) -- (3,0);
\end{tikzpicture}
\end{align}
\end{subequations}

Now use the first relation
to rewrite the action of the transfer matrix $\poly T(z)$ on $\Psi(z)$ in a similar
way as in \eqref{eq:RWtrick}:
\[
\begin{tikzpicture}[baseline=-3pt,scale=0.75,rounded corners,rotate=180]
\draw[ultra thick,decoration={markings, mark = at position 0.2 with {\arrow{>}}},postaction={decorate}] (0,0) node[thin,circle,draw,fill=white,inner sep=2pt] {} -- (2,0) (3,0) -- (5,0) node[thin,circle,draw,fill=white,inner sep=2pt] {};
\draw[ultra thick,dotted] (2,0) -- (3,0);
\draw[decoration={markings, mark = at position 0.5 with {\arrow{>}}},postaction={decorate}] (1,0) -- node[left] {$\ss z_L$} (1,-1.5);
\draw[decoration={markings, mark = at position 0.5 with {\arrow{>}}},postaction={decorate}] (4,0) -- node[left] {$\ss z_1$} (4,-1.5);
\draw[decoration={markings, mark = at position 0.4 with {\arrow{>}}},postaction={decorate}] (2.5,-1) -- (5.5,-1) -- (5.5,0.6) -- node[circle,inner sep=0.05cm,fill=black] {} (-0.5,0.6) -- node[right] {$\ss z$} (-0.5,-1) -- cycle;%
\end{tikzpicture}
\propto
\begin{tikzpicture}[baseline=-3pt,scale=0.75,rounded corners,rotate=180]
\draw[ultra thick,decoration={markings, mark = at position 0.25 with {\arrow{>}}},postaction={decorate}] (0,0) node[thin,circle,draw,fill=white,inner sep=2pt] {} -- (2,0) (3,0) -- (7,0) node[thin,circle,draw,fill=white,inner sep=2pt] {};
\draw[ultra thick,dotted] (2,0) -- (3,0);
\draw[decoration={markings, mark = at position 0.5 with {\arrow{>}}},postaction={decorate}] (1,0) -- node[left] {$\ss z_L$} (1,-1.5);
\draw[decoration={markings, mark = at position 0.5 with {\arrow{>}}},postaction={decorate}] (4,0) -- node[left] {$\ss z_1$} (4,-1.5);
\draw[decoration={markings, mark = at position 0.5 with {\arrow{>}}},postaction={decorate}] (6,0) -- (6,-1) -- (7.5,-1) -- (7.5,0.6) -- 
node[circle,inner sep=0.05cm,fill=black] {} 
(-0.5,0.6) -- node[right] {$\ss z$} (-0.5,-1) -- (5,-1) -- (5,0);%
\end{tikzpicture}
\]

Applying the $q$KZ equation \eqref{eq:qKZgraph} at $i=L$ and ignoring the spectral parameter shift since $s=1$, we find:
\begin{align*}
\begin{tikzpicture}[baseline=-3pt,scale=0.75,rounded corners,rotate=180]
\draw[ultra thick,decoration={markings, mark = at position 0.2 with {\arrow{>}}},postaction={decorate}] (0,0) node[thin,circle,draw,fill=white,inner sep=2pt] {} -- (2,0) (3,0) -- (5,0) node[thin,circle,draw,fill=white,inner sep=2pt] {};
\draw[ultra thick,dotted] (2,0) -- (3,0);
\draw[decoration={markings, mark = at position 0.5 with {\arrow{>}}},postaction={decorate}] (1,0) -- node[left] {$\ss z_L$} (1,-1.5);
\draw[decoration={markings, mark = at position 0.5 with {\arrow{>}}},postaction={decorate}] (4,0) -- node[left] {$\ss z_1$} (4,-1.5);
\draw[decoration={markings, mark = at position 0.4 with {\arrow{>}}},postaction={decorate}] (2.5,-1) -- (5.5,-1) -- (5.5,0.6) --
node[circle,inner sep=0.05cm,fill=black] {} (-0.5,0.6) -- node[right] {$\ss z$} (-0.5,-1) -- cycle;%
\end{tikzpicture}
&\propto
\begin{tikzpicture}[baseline=-3pt,scale=0.75,rounded corners,rotate=180]
\draw[ultra thick,decoration={markings, mark = at position 0.25 with {\arrow{>}}},postaction={decorate}] (0,0) node[thin,circle,draw,fill=white,inner sep=2pt] {} -- (2,0) (3,0) -- (7,0) node[thin,circle,draw,fill=white,inner sep=2pt] {};
\draw[ultra thick,dotted] (2,0) -- (3,0);
\draw[decoration={markings, mark = at position 0.5 with {\arrow{>}}},postaction={decorate}] (1,0) -- node[left] {$\ss z_L$} (1,-1.5);
\draw[decoration={markings, mark = at position 0.5 with {\arrow{>}}},postaction={decorate}] (4,0) -- node[left] {$\ss z_1$} (4,-1.5);
\draw[decoration={markings, mark = at position 0.5 with {\arrow{>}}},postaction={decorate}] (5,0) -- (5,-1) -- node[above] {$\ss z$} (6,-1) -- (6,0);
\end{tikzpicture}
\\
&\propto
\begin{tikzpicture}[baseline=-3pt,scale=0.75,rounded corners,rotate=180]
\draw[ultra thick,decoration={markings, mark = at position 0.35 with {\arrow{>}}},postaction={decorate}] (0,0) node[thin,circle,draw,fill=white,inner sep=2pt] {} -- (2,0) (3,0) -- (5,0) node[thin,circle,draw,fill=white,inner sep=2pt] {};
\draw[ultra thick,dotted] (2,0) -- (3,0);
\draw[decoration={markings, mark = at position 0.5 with {\arrow{>}}},postaction={decorate}] (1,0) -- node[left] {$\ss z_L$} (1,-1.5);
\draw[decoration={markings, mark = at position 0.5 with {\arrow{>}}},postaction={decorate}] (4,0) -- node[left] {$\ss z_1$} (4,-1.5);
\end{tikzpicture}
\end{align*}
which is the desired eigenvector property.

\section{Properties of the spin-loop mapping}\label{app:spinloop}
We discuss here in more detail the linear map $\varphi_\ell$ from the loop space to the spin space that is used in section \ref{sec:loop}.

\subsection{Case \texorpdfstring{$\ell=1$}{l=1}}
This case is well-known, and we only review it briefly.
$\varphi_1=\varphi: \mathbb{C}[\mathcal{L}_L]\to (\mathbb{C}^2)^{\otimes L}$ 
is defined explicitly in \ref{sec:loopequiv}. Both spaces
come equipped with standard bases, indexed by $\mathcal{L}_L$ and $\{0,1\}^L$
respectively, and its matrix is given by \eqref{eq:spinloop};
its entries are powers of $-q$, and therefore $\varphi$
is well-defined for any $q\in\mathbb{C}^\times$.
Up to an overall power of $q$, these are nothing but maximal parabolic
Kazhdan--Lusztig polynomials, see e.g.~\cite{artic52} for a discussion
in a related context.

There is a partial order on $\{0,1\}^L$:
$\epsilon\le\epsilon'$ iff $\sum_{j=1}^i (\epsilon_j-\epsilon'_j)\le 0$
for all $i=1,\ldots,L$. Denote $\epsilon_0=\{1,0,\ldots,1,0\}$
and $\tilde{\mathcal{L}}_L=\{\epsilon\in\{0,1\}^L: \epsilon\ge\epsilon_0\}$
(Dyck words).
There is a 
natural bijection between $\mathcal{L}_L$ and $\tilde{\mathcal{L}}_L$ 
which to $\pi\in\mathcal{L}_L$ associates $\epsilon$ such that
\[
\epsilon_i=\begin{cases}
1&\text{if $\pi(i)>i$}\\
0&\text{if $\pi(i)<i$}
\end{cases}
\qquad
\qquad i=1,\ldots,L
\]

The claim is that the submatrix of $\phi$ where indices
are restricted to $\tilde{\mathcal{L}}_L
\times\mathcal{L}_L$, is upper triangular w.r.t.\ the partial order above
(modulo the bijection). This is fairly obvious in the graphical interpretation
where one orients the arcs: the diagonal elements of the matrix
correspond to all arcs pointing right, and each subsequent change of orientation
of an arc clearly increases the partial sums $\sum_{j\le i}\epsilon_i$.
Furthermore, the diagonal elements are all $(-q)^{-L/2}$.
Therefore, the matrix has maximal rank for all $q\in\mathbb{C}^\times$,
and $\varphi$ is an isomorphism of $\mathbb{C}[\mathcal{L}_L]$ onto its image.

For $L=2$, the unique link pattern is sent onto the state
$(-q)^{-1/2}\ketE{10}+(-q)^{1/2}\ketE{01}$, which is checked explicitly
to be $U_1$-invariant. One can then proceed inductively by noting that
any link pattern
possesses an arc connecting neighbors, applying the $L=2$ observation to it
and then removing it. Therefore all link patterns are sent into the 
$U_1$-invariant subspace $((\mathbb{C}^2)^{\otimes L})_{inv}$
of $(\mathbb{C}^2)^{\otimes L}$.

For generic $q$, the dimension of the invariant subspace
is the same as for $\mathfrak{sl}(2)$ and is well-known to be the
Catalan number $L!/((L/2)!(L/2+1)!)$ ($L$ even), which enumerates Dyck paths
or link patterns. Therefore, for generic $q$,
$\varphi$ is an isomorphism of $\mathbb{C}[\mathcal{L}_L]$ onto 
$((\mathbb{C}^2)^{\otimes L})_{inv}$.

\subsection{Fused case}\label{app:loopfused}
Let us first define the projection map $p=p^{(\ell)}:
(\mathbb{C}^2)^{\otimes \ell}\to V\cong \mathbb{C}^{\ell+1}$.
We use induction on $\ell$: $p^{(1)}=1$ and
\[
p^{(k+1)}=p^{(k)}\left(1+\frac{q^{k}-q^{-k}}{q^{k+1}-q^{-(k+1)}}e_k\right)p^{(k)}
\]
where $e_k=\left(\begin{smallmatrix}0&0&0&0\\ 0&-q&1&0\\ 0&1&-q^{-1}&0\\ 0&0&0&0
	\end{smallmatrix}\right)_{k,k+1}$; in the loop language, $e_k$ is simply 
\linkpattern[tangle,numbering={1/k,1'/k,2/k+1,2'/k+1}]{1/2,1'/2'}.
From its definition, it commutes with the $U_1$-action, and in fact the claim
(see e.g.~\cite{Martin-Potts}) 
is that it is exactly the projection onto the maximal spin $\ell/2$ 
irreducible subrepresentation.

Now define $\varphi_\ell$ from $\varphi$ by projection, as in \eqref{eq:projfus}. It is well-defined when the projection is, that is, when $q^{2r}\ne 1$, $r=2,\ldots,\ell$.
By definition, the
$\varphi_\ell(\pi)$ live in the $U_1$-invariant 
subspace $(V^{\otimes L})_{inv}$ of $V^{\otimes L}$.
Furthermore, borrowing an argument from \cite{artic37}, we note
that linear independence of the $\varphi_\ell(\pi)$ is equivalent
to linear independence of the $\varphi(\pi)$, because
$\varphi_\ell(\pi)=\varphi(\pi)+\sum_{\rho\in\mathcal{L}_{\ell L}-\mathcal{L}_{\ell;L}} c_{\pi,\rho} \varphi(\rho)$. But the latter was proved in the previous section.

It is not hard to show once again 
that $\dim (V^{\otimes L})_{inv}=\# \mathcal{L}_{\ell,L}$
for generic $q$
(a short proof is to use a bijection -- generalizing the one of the
previous section --
between fused link patterns and Bratteli diagrams for spin $\ell/2$,
the increments of the paths being now given as the difference of numbers
of opening and closing arcs in a given group); 
equivalently, using the crystal limit $q\to 0$
one can conclude directly that the map above is an isomorphism between
$\mathbb{C}[\mathcal{L}_{\ell,L}]$ and $(V^{\otimes L})_{inv}$ for generic $q$.
We conjecture that it remains an isomorphism
for
$q=-e^{-i\pi/(\ell+2)}$, though we do not have a rigorous proof of that.

\subsection{Wheel condition}\label{sec:proofwheel}
\let\up=\uparrow
\let\down=\downarrow
\let\ket=\ketC
\let\bra=\braC
We prove here the wheel condition of Theorem~\ref{thm:wheel}.
\rem{does the proof work for any $k$? is $\Psi^{(k)}$ $U_1$-invariant??}

We work inside $(\mathbb{C}^2)^{\otimes \ell L}$.
We fix a subset of consecutive $2\ell$ indices,
say $\ell(i-1)+1,\ldots,\ell (i+1)$, 
and consider the projector $P_{j;i,i+1}$ ($j=0,\ldots,\ell$) 
onto the spin $j$ subrepresentation
of the action of $U_q(\mathfrak{sl}(2))$ acting on these $2\ell$ spaces.
Also, given a link pattern $\pi\in\mathcal{L}_{L,\ell}$.
denote $2c_{i,i+1}(\pi)$ the number of connections between the 
sites $\ell(i-1)+1,\ldots,\ell(i+1)$ and the outside.

Then we have the following lemma:
$j>c_{i,i+1}(\pi)$ implies $P_{j;i,i+1}\varphi_1(\pi)=0$ 
(and the $P_{j;i,i+1} \varphi_1(\pi)$, where $\pi$ runs over
link patterns with $c_{i,i+1}=j$, are linearly independent).
The equality is easy proved by inductively removing the ``little arches''
connecting sites $\ell(i-1)+1,\ldots,\ell(i+1)$. 
Explicitly, suppose $\pi(k)=k+1$,
$\ell(i-1)+1\le k<k+1\le \ell(i+1)$, and consider the vector space
\begin{multline*}
V=(\mathbb C^2)_{\ell(i-1)+1}\otimes\cdots
\otimes (\mathbb C^2)_{k-1}\otimes 
((-q)^{-1/2} \ket{\up\down}-(-q)^{1/2}\ket{\down\up})_{k,k+1}\\\otimes
(\mathbb C^2)_{k+2}\otimes\cdots\otimes (\mathbb C^2)_{\ell(i+1)}
\end{multline*}
Note that $\varphi_1(\pi)\in W_{\text{left}} \otimes V \otimes W_{\text{right}}$,
where $W_{\text{left}}\cong(\mathbb{C}^2)^{\otimes \ell(i-1)}$ 
and $W_{\text{right}}\cong(\mathbb{C}^2)^{\otimes \ell(L-i-1)}$ 
correspond to
the sites outside $[\ell(i-1)+1,\ell(i+1)]$.
Since the vector at sites $k,k+1$ is a singlet of $U_q(\mathfrak{sl}(2))$,
$V$ is isomorphic to $(\mathbb{C}^2)^{\otimes 2(\ell-1)}$ as
a representation space of $U_q(\mathfrak{sl}(2))$.
Iterating the process, we find that
$\varphi_1(\pi)\in W_{\text{left}}\otimes V\otimes W_{\text{right}}$ as above, where $V$ is a subspace of
$(\mathbb C^2)^{\otimes 2\ell}$ which is
isomorphic to $(\mathbb C^2)^{\otimes 2c_{i,i+1}(\pi)}$. But we know explicitly
the decomposition of the latter into irreducible representations,
and in particular that the highest possible spin is $c_{i,i+1}(\pi)$, 
hence the equality of the lemma. 
For the linear independence, consider the unique spin state obtained
from $\pi$ which 
has all outside connections of the form $\ket{\uparrow}$, i.e.,
by orienting all lines to the outside ``outwards''. 

Next, we note that since the projection
$\prod_{i=1}^L p_i$ {\em cannot increase}\/ the number of connections
to the outside, then the same statement can be made about 
the ``fused'' vector $\varphi_\ell(\pi)$, $\pi\in\mathcal{L}_{L,\ell}$
(noting that in the case of fused link patterns, the connections to
the outside are shared equally between 
groups $i$ and $i+1$).
So,
\[
P_{j;i,i+1}\varphi_\ell(\pi)=0\qquad
\pi\in\mathcal{L}_{L;\ell},\ 
j>c_{i,i+1}(\pi)
\]
(and the $P_{j;i,i+1} \varphi_1(\pi)$, where $\pi$ runs over
fused link patterns with $c_{i,i+1}=j$, are linearly independent).

From the coproduct of $U_q(\mathfrak{sl}(2))$ it not hard
to see that $P_{j;i,i+1}$ commutes with its action {\em on the whole of}\/
$(\mathbb{C}^{\ell+1})^{\otimes L}$, and in particular leaves invariant
the image of $\varphi_\ell$. This implies that
$P_{j;i,i+1} \varphi_\ell(\pi)$ is a linear combination of $\varphi_\ell(\pi')$.

Now writing $P_{\ell;i,i+1} P_{j;i,i+1} \varphi_\ell(\pi)=0$ for $j<\ell$, and expanding 
$P_{j;i,i+1}\varphi_\ell(\pi)$ 
in $\varphi_\ell(\pi')$ as above, we find that no link patterns
$\pi'$ with $c_{i,i+1}(\pi')=\ell$ can appear in the expansion.
Repeating the argument with $P_{\ell-1;i,i+1}P_{j;i,i+1}\varphi_\ell(\pi)=0$, etc,
until $P_{j+1;i,i+1}P_{j;i,i+1}\varphi_\ell(\pi)=0$, we conclude that
\[
\mathrm{Im}\,P_{j;i,i+1}\circ\varphi_\ell \subseteq \left<\varphi_\ell(\pi): c_{i,i+1}(\pi)\le j\right>
\]

Now consider
$\poly\Psi:=\poly\Psi^{(0)}(\ldots,z_i,z_{i+1},z_{i+2},\ldots)$ (the case $k>0$ will be discussed separately)
with
$z_i=q^{2(a+b)}z$, $z_{i+1}=q^{2b}z$, $z_{i+2}=z$, $a,b>0$, $a+b<\ell+2$.
Plugging these equalities into
\eqref{eq:exchp} and using the form \eqref{eq:defRp} of the
$R$-matrix,
we conclude that $\poly\Psi$ is a linear combination
of $P_{j;i,i+1} \poly\Psi$, $j<a$, and a linear combination of
$P_{j;i+1,i+2} \poly\Psi$, $j<b$. 

As mentioned in sect.~\ref{sec:loop}, $\poly\Psi$ is $U_1$-invariant.
We can therefore apply the reasoning above and conclude that
\[
\poly\Psi
\in \left<\varphi_\ell(\pi): c_{i,i+1}(\pi)< a,\,
c_{i+1,i+2}(\pi)< b\right>
\]
For any such $\pi$ the group of sites $i+1$ has 
at least $\ell+1-a$ connections
to the sites $i$, and $\ell+1-b$ connections to the sites $i+2$;
but $\ell+1-a+\ell+1-b>\ell$ which is contradictory.
Therefore $\poly\Psi=0$.

This is the special case $i_1=i$, $i_2=i+1$, $i_3=i+2$ of \eqref{eq:wheel}. 
In order to obtain
the general case, one simply needs to apply the exchange relation
\eqref{eq:exchp} repeatedly, to move $(i,i+1,i+2)\to (i_1,i_2,i_3)$;
assuming that all other arguments of $\poly\Psi$ are generic,
all $R$-matrices involved in this process have nonzero denominator
and are therefore well-defined.
We have thus proved
Theorem~\ref{thm:wheel}.

The case of $\Psi^{(k)}$, $k>0$, can be treated similarly by noting that the reasoning above
is purely local and cannot depend on the twist. Explicitly, one can build
a $U_1$-invariant element of $\mathbb C^{k+1}\otimes(\mathbb C^2)^{\otimes\ell L}\otimes\mathbb C^{k+1}$
given by $\bra{x}\Phi(z_1)\ldots\Phi(z_{2n})\ket{y}$ where $\bra{x}$ and $\ket{y}$ run
over the $U_1$-representation generated by $\bra{k}$ and $\ket{k}$, each isomorphic to $\mathbb C^{k+1}$.
It satisfies the wheel condition, and therefore $\Psi^{(k)}$ does.

{\em Remark}: the connection between the projectors
$P_j$ and link patterns can be made more explicit 
by introducing generalized Temperley--Lieb
operators $E_j$ (denoted $e^{(\ell-j)}$ in \cite{artic37}),
which are best described graphically as
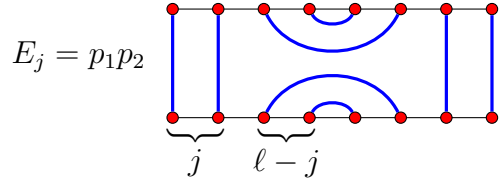
\[
E_j= p_1 p_2\ 
\begin{tikzpicture}[/linkpattern/every linkpattern]
\linkpattern[tangle,numbered=false,tikzstarted,height=1.2]{1/1',2/2',3/6,4/5,3'/6',4'/5',7/7',8/8'}
\node[rotate=90] at (1.5,-1.6) {$\Big\{$};
\node[rotate=90] at (3.5,-1.6) {$\Big\{$};
\node at (1.5,-2.2) {$j$};
\node at (3.5,-2.2) {$\ell-j$};
\end{tikzpicture}
\]
Note that the $E_j$ are (up to normalization) a family of 
non-orthogonal projectors (as opposed to the $P_j$ which are orthogonal
to each other).

Comparing the expression of the $R$-matrix
\eqref{eq:defRp}
 with the following expression (correcting a sign mistake in Eq.~(2.8) of \cite{artic37}):
\[
\check {\poly R}(z)=(-1)^\ell
\sum_{j=0}^\ell
\left(\prod_{r=1}^j \frac{q^{\ell-r+1}-q^{-(\ell-r+1)}}{q^r-q^{-r}}\right)
\left(\prod_{r=0}^j
\frac{q^r z-q^{-r}w}{q^{\ell-r}w-q^{r-\ell}z}
\right)
E_j
\]
we see that there is a triangular change of basis between the
$P_j$ and the
$E_j$: $P_j=\sum_{j'\le j} \alpha_{j,j'} E_{j'}$ with $\alpha_{j,j}\ne0$.

\section{Polynomiality}\label{app:poly}
In this appendix we prove that the expression~\eqref{eq:Psi_pol} is a polynomial.
As a bonus we get that the coefficients are Laurent polynomials in $q$ up to some known factor.
Recall the expression:
\begin{equation}\label{eq:recall_Psi}
 \begin{aligned}
 \poly\Psi_{b_1, \ldots, b_{2n}}^{(k)} (z_1,\ldots,z_{2n}) &=\gamma_{\ell,n}^{(k)} \prod_{i=1}^{2n} \alpha_{b_i} z_i^k
\prod_{1\le i<j\le 2n} \prod_{r=1}^\ell (q^{2r}z_j-z_i)
\oint \ldots \oint \prod_{i=1}^{\ell n}\frac{dw_i}{2\pi i}\\
 & \quad \times \prod_{i<j}\frac{w_j-w_i}{w_j-q^2 w_i} \frac{P_{Y_{\ell,n}^{(k)}} (w_1,\ldots, w_{\ell n})}{\prod_{j \leq \e_i} (q^{\ell} w_i-z_j) \prod_{j \geq \e_i}(w_i - q^{\ell}z_j)} 
 \end{aligned}
\end{equation}
where the contour integration is made in the following way: we start by integrating $w_{\ell n}$ arround the points $q^{\ell-2k} z_s$, for all $1\leq i\leq 2n$ and $0\leq k< \ell$; we proceed by doing the same with $w_{\ell n-1}$, $w_{\ell n-2}$, \ldots, up to $w_1$.
By Cauchy's residue theorem this is the same as replacing $\{w_1, \ldots,w_{\ell n}\}\to \{q^{\ell-2k_1} z_{s_1},\ldots,q^{\ell-2k_{\ell n}} z_{s_{\ell n}}\}$, removing the correspondant poles $(w_i-q^{\ell-2k_i}z_{s_i})^{-1}$ and summing over all possible combinations.
Notice that, due to the presense of the Vandermonde polynomial we can not repeat a pole. 
Also, $q^{\ell-2k}z_s$ is only possible, for $k\geq 1$ if $q^{\ell-2k+2}z_s$ appears to the right.

We want to prove that there is neither any pole of the kind $(z_s - q^a z_r)^{-1}$ nor of the kind $z_r^{-1}$.
In order to prove it we will look in the details of the computation.
We can ignore all integrations except the ones that concern $z_r$ and $z_s$, with $r<s$.
There is a subset of the integration variables $\{w_{j_1},\ldots,w_{j_m}\}$ which is replaced by $q^{\ell -2k} z_r$ and $q^{\ell-2k} z_s$, the others being replaced by other variables $z_t$.

Fix $m$.
We will represent it diagramatically, representing all poles in $z_s$ by a blue circle and the ones in $z_r$ by a red square.
For example, we represent
\begin{align*}
 \{w_1,w_2,w_3,w_4\} &\to \{q^{\ell-4} z_s, q^{\ell-2} z_s,q^{\ell} z_r,q^{\ell} z_s\} &\text{by}&\quad
\begin{tikzpicture}[scale=.6]
 \draw[black!50!white] (0,0) -- (3,0);
 \node[Ps] at (0,0) {};
 \node[Ps] at (1,0) {};
 \node[Pr] at (2,0) {};
 \node[Ps] at (3,0) {};
\end{tikzpicture}
\end{align*}
where we simplified the indices of the integration variables.

Let $\psi_{d, \bm \e}$ be the relevant part of the integral, where $d$ stands for an arbitrary diagram with $m$ circles and squares,
\begin{equation}\label{eq:form_poly}
 \begin{aligned}
 \psi_{d,\bm \e} (z_r,z_s) &=
\oint \ldots \oint \prod_{i=1}^{m}\frac{dw_i}{2\pi i}\prod_{i<j}\frac{w_i-w_j}{w_i-q^{-2} w_j} 
 \prod_i \frac{1}{(w_i - q^{\ell}z_r)(w_i - q^{\ell}z_s)}\\ 
 & \quad \times P_{Y_{\ell,n}^{(k)}} z_r^k z_s^k \prod_{i=1}^m \frac{\prod_{\e_i > r}(w_i - q^{\ell}z_r)\prod_{\e_i > s}(w_i - q^{\ell}z_s)}{\prod_{\e_i \geq r} (q^{\ell} w_i-z_r) \prod_{\e_i \geq s} (q^{\ell} w_i-z_s) } 
\prod_{t=1}^\ell (q^{2t}z_s-z_r)
 \end{aligned}
\end{equation}
where the integration is performed according to the diagram $d$.
We restrict $\bm \e$ to the relevant indices: $\{\e_{j_1}, \ldots \e_{j_m}\}$.

The following statement is necessary to the proof:
\begin{lemma}\label{lemma:Macdonald_power}
 The Macdonald polynomial $P_{Y_{\ell,n}^{(k)}}$ satisfies:
\[
  P_{Y_{\ell,n}^{(k)}} (z,q^2 z, \ldots, q^{2j-2} z, w_{j+1}, \ldots) = z^{\max{\{j-k,0\}}} \times (\ldots)
\]
where $j\leq \ell$.
\end{lemma}

Recall that $P_{Y_{\ell,n}}^{(k)}$ is obtained from $\left(\prod_i w_i\right) \braC{k_1, \ldots, k_{\ell}} e^{2\ell n \bm \beta} \mathcal F \ketC{k_1, \ldots, k_{\ell}}$, 
which can be seen as a sum over all parafermionic decompositions.
Chosing $j$ variables $\{z,q^2 z, \ldots, q^{2j-2}z\}$ forces the parafermions to be in distinct modes, according to~\eqref{eq:pf_f_relations}.
Each parafermion $\bm f^i (z)$ generates a factor $z^{-k_i}$.
The result follows from minimizing the power of $z$ on the expression.

We split formula~\eqref{eq:form_poly} into two contributions, the first row and the second row.
All the poles of the second row are outside of the contour integral.
\begin{proposition}\label{prop:poly_2nd_term}
 For any diagram $d$, the second row of formula~\eqref{eq:form_poly} is a polynomial in $z_r$ and $z_s$.
 Moreover, if multiplied by $\frac{[\ell]!}{[\ell-m_r]!}\frac{[\ell]!}{[\ell-m_s]!}$, where $m_r$ (and $m_s$) is the number of $\e_i$ equal to $r$ (respectivelly $s$), it is a Laurent polynomial in $q$.
\end{proposition}

Notice that $m_r \leq b_r$ (and $m_s \leq b_s$), the equality being only possible if we chose all $w_i$ such that $\e_i = r$ (or $s$) to belong to our subset of variables.

\begin{proof}
 By choice $r<s$.
 Pick a random diagram, for example
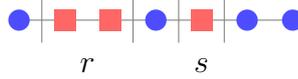
\begin{center}
 \begin{tikzpicture}[scale=.6]
 \draw[black!50!white] (0,0) -- (6,0);
 \node[Ps] at (0,0) {};
 \draw[black!50!white] (0.5,-0.5) -- (0.5,0.5);
 \node at (1.5,-1) {$r$}; 
 \node[Pr] at (1,0) {};
 \node[Pr] at (2,0) {};
 \draw[black!50!white] (2.5,-0.5) -- (2.5,0.5);
 \node[Ps] at (3,0) {};
 \draw[black!50!white] (3.5,-0.5) -- (3.5,0.5);
 \node at (4,-1) {$s$}; 
 \node[Pr] at (4,0) {};
 \draw[black!50!white] (4.5,-0.5) -- (4.5,0.5);
 \node[Ps] at (5,0) {};
 \node[Ps] at (6,0) {};
 \end{tikzpicture}
\end{center}
where the vertical lines split the variables $w_i$ in five regions $\e_i < r$, $\e_i = r$, $r < \e_i < s$, $\e_i = s$ and $\e_i > s$.
The product $(w_i - q^{\ell} z_r)(w_i - q^{\ell} z_s)$ vanishes whenever there is a blue circle on the region $\e_i>s$; and it vanishes also if there is a red square on the regions $\e_i >r$.
Then this example will vanish.

We pick now a non vanishing example:
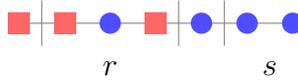
\begin{center}
 \begin{tikzpicture}[scale=.6]
 \draw[black!50!white] (0,0) -- (6,0);
 \node[Pr] at (0,0) {};
 \draw[black!50!white] (0.5,-0.5) -- (0.5,0.5);
 \node at (2,-1) {$r$}; 
 \node[Pr] at (1,0) {};
 \node[Ps] at (2,0) {};
 \node[Pr] at (3,0) {};
 \draw[black!50!white] (3.5,-0.5) -- (3.5,0.5);
 \node[Ps] at (4,0) {};
 \draw[black!50!white] (4.5,-0.5) -- (4.5,0.5);
 \node at (5.5,-1) {$s$}; 
 \node[Ps] at (5,0) {};
 \node[Ps] at (6,0) {};
 \end{tikzpicture}
\end{center}
The term $(q^{\ell}w_i - z_s)^{-1}$ only appears in the region $\e_i=s$, exacly $m_s$ times, when we replace $w_i$ by the sucessive $q^{\ell-2p} z_s$ we get $(q^{2\ell}-1)^{-1}\ldots (q^{2\ell-2m_s+2} - 1)^{-1} z_s^{-m_s}$.
There are at least $m_s$ variables $z_s$ in $P_{Y_{\ell,n}^{(k)}}$ and we can apply lemma~\ref{lemma:Macdonald_power}, therefore the power in $z_s$ is at least $z^k_s z^{\max{\{m_s-k,0\}}}_s z^{-m_s}_s = z_s^{\max{\{0,k-m_s\}}}$.
The term $(q^{\ell}w_i - z_r)^{-1}$ appears in all regions such $\e_i \geq r$.
Let $m^{\prime}_r$ be the number of red squares in the region $\e_i = r$.
Then we can do the same analysis for the squares. 
The circles will generate a term $(q^{2\ell}z_s -z_r)\ldots(q^{2\ell-2m^{\prime}+2} z_s-z_r) $, where $m^{\prime}$ is the number of circles in the regions $\e_i \geq r$.
This term cancels with the product $\prod_p (q^{2p}z_s-z_r)$.
\end{proof}

In general, the first row of the formula~\eqref{eq:form_poly} will give a rational function.
\begin{proposition}\label{prop:poly_1st_term}
 The sum over all possible diagrams with $m$ nodes of $\psi_{d,\bm \e} (z_r, z_s)$ is a polynomial in $z_r$ and $z_s$.
\end{proposition}

\begin{proof}
From proposition~\ref{prop:poly_2nd_term} we know that
\[
 \psi_{d} (z_r,z_s) =
\oint \ldots \oint \prod_{i=1}^{m}\frac{dw_i}{2\pi i}\prod_{i<j}\frac{w_j-w_i}{w_j-q^2 w_i} 
 \prod_i \frac{Q_d (z_r,z_s)}{(w_i - q^{\ell}z_r)(w_i - q^{\ell}z_s)} 
\]
where $Q_d (z_r, z_s)$ is a polynomial depending on the diagram $d$.
We ignore the dependence in $\bm \e$.
For convenience we allow diagrams with more than $\ell$ circles (or squares), in which case we set $Q_d$ to zero.
This is naturally done by using the Macdonald Polynomial $P_{Y_{\ell,n}}$.

On the other hand, $Q_d (z_r,z_s)$ can be seen as the evaluation of a function, i.e. $Q_d (z_r,z_s) = Q (z_r,z_s;z^{(1)},\ldots,z^{(m)})$, where $z^{(i)}$ is given by the diagram $d$.
We will omit the explicit dependence on $z_r$ and $z_s$, and for the analysis we will consider that the $z^{(i)}$ are independent variables.
Then, this is a rational function of the form
\begin{equation}\label{eq:rational_Q}
  Q(z^{(1)},\ldots,z^{(m)}) = \frac{R(z^{(1)},\ldots,z^{(m)})}{\prod_i \prod_{j_i}(z^{(i)}-\zeta_{j_i})}
\end{equation}
where the polynomial $R$ and $\zeta_{j_i}$ depend on $z_r$ and $z_s$.

The diagram \tikz[scale=.6]{
 \draw[black!50!white] (0,0) -- (3,0);
 \node[Ps] at (0,0) {};
 \node[Pr] at (1,0) {};
 \node[Ps] at (2,0) {};
 \node[Ps] at (3,0) {};}
corresponds to the following expression:
\[
 \psi_{d}(z_r, z_s) = \frac{Q(q^{\ell-4}z_s,q^{\ell}z_r,q^{\ell-2}z_s,q^{\ell}z_s)}{(q^{\ell-4}z_s - q^{\ell-2}z_r)(q^{\ell}z_r - q^{\ell-4}z_s)(q^{\ell-2}z_s - q^{\ell}z_r)(q^{\ell}z_s - q^{\ell}z_r)} 
\]
The denominator can be constructed by a simple rule:
we read the diagram $d$ from right to left, at each node we can put either a circle or a square (corresponding to some $q^a z_s$ or $q^b z_r$); if we choose a circle (resp.\ a square) we add the pole $(q^a z_s - q^b z_r)^{-1}$ (resp.\ $(q^b z_r - q^a z_s)^{-1}$).

Now we sum over all $2^m$ possible diagrams.
Let $d_1$ and $d_2$ be two diagrams that only differ in the first element, i.e. $d_1 = \tikz[baseline=-3pt]{\node[pr] at (0,0) {};}\,d^{(1)}$ and $d_2 = \tikz[baseline=-3pt]{\node[ps] at (0,0) {};}\,d^{(1)}$, where $d^{(1)}$ is some diagram with $m-1$ nodes.
Then:
\[
 \psi_{\tikz[baseline=-3pt]{\node[pr] at (0,0) {};}\,d^{(1)}}(z_r,z_s) + 
\psi_{\tikz[baseline=-3pt]{\node[ps] at (0,0) {};}\,d^{(1)}} (z_r,z_s) =
\frac{1}{(\ldots)} \frac{Q_{\tikz[baseline=-3pt]{\node[pr] at (0,0) {};}\,d^{(1)}}-Q_{\tikz[baseline=-3pt]{\node[ps] at (0,0) {};}\,d^{(1)}}}{q^a z_r - q^b z_s}
\]
where the three dots are replacing the commun $m-1$ poles.
Using the formulation~\eqref{eq:rational_Q} and setting $q^a z_r = z^{(1)}$ and $q^b z_s = z^{\prime(1)}$, the numerator is given by
\[
 \frac{1}{\prod_{i\geq 2} \prod_{j_i} (z^{(i)}-\zeta_{j_i})} \left(\frac{R(z^{(1)},\ldots)}{\prod_{j_1}( z^{(1)}-\zeta_{j_1})}-\frac{R(z^{\prime(1)},\ldots)}{\prod_{j_1} (z^{\prime(1)}-\zeta_{j_1})}\right)
\]
which vanishes if $z^{(1)}=z^{\prime (1)}$.
Then the division by $(z^{(1)}-z^{\prime (1)})$ is well defined and we can define the polynomial:
\[
Q^{(1)}_{\tikz[baseline=-3pt]{\node[pn] at (0,0) {};}\,d^{(1)}} \coleq \frac{Q_{\tikz[baseline=-3pt]{\node[pr] at (0,0) {};}\,d^{(1)}}-Q_{\tikz[baseline=-3pt]{\node[ps] at (0,0) {};}\,d^{(1)}}}{q^a z_s - q^b z_r}
\]
where the open circle indicates that we already performed the sum.
This can be written in the rational function perpesctive:
\[
  Q^{(1)} (z^{(1)},z^{\prime(1)};z^{(2)}\ldots,z^{(m)}) = \frac{R^{(1)}(z^{(1)},z^{\prime(1)};z^{(2)}\ldots,z^{(m)})}{\prod_k(z^{(1)}-\zeta_k)(z^{\prime(1)}-\zeta_k) \prod_{i\geq 2} \prod_{j_i}(z^{(i)}-\zeta_{j_i})}
\]
Notice that this is a symmetric function on $z^{(1)}$ and $z^{\prime (1)}$.

Now we sum over the second node, let $d_1 = \tikz[baseline=-3pt]{\node[pn] at (0,0) {};\node[pr] at (0.2,0) {};}\,d^{(2)}$ and $d_2 = \tikz[baseline=-3pt]{\node[pn] at (0,0) {};\node[ps] at (0.2,0) {};}\,d^{(2)}$, where the open circle indicates the first node where we already performed the sum.
Then:
\[
 \psi_{\tikz[baseline=-3pt]{\node[pn] at (0,0) {}; \node[pr] at (0.2,0) {};}\,d^{(2)}}(z_r,z_s) + 
\psi_{\tikz[baseline=-3pt]{\node[pn] at (0,0) {}; \node[ps] at (0.2,0) {};}\,d^{(2)}} (z_r,z_s) =
\frac{1}{(\ldots)} \frac{Q^{(1)}_{\tikz[baseline=-3pt]{\node[pn] at (0,0) {};\node[pr] at (0.2,0) {};}\,d^{(2)}}-Q^{(1)}_{\tikz[baseline=-3pt]{\node[pn] at (0,0) {};\node[ps] at (0.2,0) {};}\,d^{(2)}}}{q^a z_s - q^b z_r}
\]
We consider now that $z^{(1)}$ and $z^{\prime(1)}$ depend on the two variables $z^{(2)}$ and $z^{\prime(2)}$, then the numerator is given by:
\begin{multline*}
 \frac{1}{\prod_{i\geq 3} \prod_{j_i} (z^{(i)}-\zeta_{j_i})} \left(\frac{R^{(1)}(q^{-2}z^{(2)},z^{\prime(2)};z^{(2)},\ldots)}{\prod_k(q^{-2} z^{(2)}-\zeta_k)(z^{\prime(2)}-\zeta_k)\prod_{j_2}(z^{(2)}-\zeta_{j_2})}\right.\\
-\left.\frac{R^{(1)}(z^{(2)},q^{-2} z^{\prime(2)};z^{\prime(2)},\ldots)}{\prod_k (z^{(2)}-\zeta_k)(q^{-2}z^{\prime(2)}-\zeta_k)\prod_{j_2}(z^{\prime (2)}-\zeta_{j_2})}\right)
\end{multline*}
Once again this vanishes if $z^{(2)} = z^{\prime(2)}$, then the division by $(z^{(2)}-z^{\prime(2)})$ is well defined and we can define a new polynomial:
\[
Q^{(2)}_{\tikz[baseline=-3pt]{\node[pn] at (0,0) {};\node[pn] at (0.2,0) {};}\,d^{(2)}} \coleq \frac{Q^{(1)}_{\tikz[baseline=-3pt]{\node[pn] at (0,0) {};\node[pr] at (0.2,0) {};}\,d^{(2)}}-Q^{(1)}_{\tikz[baseline=-3pt]{\node[pn] at (0,0) {};\node[ps] at (0.2,0) {};}\,d^{(2)}}}{q^a z_s - q^b z_r}
\]
or in the rational function perspective:
\[
  Q^{(2)} (z^{(2)},z^{\prime(2)};z^{(3)}\ldots,z^{(m)}) = \frac{R^{(2)}(z^{(2)},z^{\prime(2)};z^{(3)}\ldots,z^{(m)})}{\prod_k(z^{(2)}-\zeta_k)(z^{\prime(2)}-\zeta_k) \prod_{i\geq 2} \prod_{j_i}(z^{(i)}-\zeta_{j_i})}
\]
which is once again symmetric in $z^{(2)}$ and $z^{\prime (2)}$.

Repeating the process $m$ times we obtain:
\[
 \psi (z_r,z_s) = Q^{(m)} (q^{\ell} z_r, q^{\ell} z_s)
\]
which is a polynomial.
\end{proof}

This concludes the proof.
Note that if $Q_d$ is a Laurent polynomial in $q$ for all $d$, then $\sum_d \psi_d$ will be also a Laurent polynomial in $q$.
Then:
\begin{corollary}
 The coefficients of $([\ell]!)^{-n}\prod_i [b_i]!\poly\Psi_{b_1, \ldots, b_{2n}}^{(k)}(z_1,\ldots,z_{2n})$ are Laurent polynomials in $q$.
\end{corollary}

\def\urlshorten http://#1/#2urlend{{#1}}%
\renewcommand\url[1]{%
\href{#1}{\scriptsize\expandafter\path\urlshorten#1 urlend}%
}
\gdef\MRshorten#1 #2MRend{#1}%
\gdef\MRfirsttwo#1#2{\if#1M%
MR\else MR#1#2\fi}
\def\MRfix#1{\MRshorten\MRfirsttwo#1 MRend}
\renewcommand\MR[1]{\relax\ifhmode\unskip\spacefactor3000 \space\fi
\MRhref{\MRfix{#1}}{{\scriptsize \MRfix{#1}}}}
\renewcommand{\MRhref}[2]{%
\href{http://www.ams.org/mathscinet-getitem?mr=#1}{#2}}

\bibliography{biblio}
\bibliographystyle{amsplainhyper}

\end{document}